\documentclass[11pt,reqno]{amsart}

\usepackage{amsmath,amsthm,amssymb,amsfonts,mathtools}
\usepackage{bm}
\usepackage{natbib}
\usepackage{bbm}
\usepackage{graphicx}
\usepackage{multirow}
\usepackage{here}
\usepackage{fullpage}
\usepackage{xurl}
\usepackage{color}
\usepackage{xcolor}
\usepackage{mathrsfs}
\usepackage{enumitem}
\usepackage{comment}
\usepackage{subcaption}
\usepackage{pifont}
\usepackage[linesnumbered,ruled,vlined]{algorithm2e}
\usepackage[hidelinks]{hyperref}
\makeatletter
\makeatletter

\newtheorem{lemma}{Lemma}[section]
\newtheorem{proposition}{Proposition}
\newtheorem{assumption}{Assumption}[section]
\theoremstyle{definition}
\newtheorem{definition}{Definition}[section]
\newtheorem{remark}{Remark}[section]
\newtheorem{example}{Example}[section]

\def\E#1{\mathbb{E}\left[#1\right]}
\def\P#1{\mathbb{P}\left[#1\right]}

\newcommand{\argmin}{\mathop{\rm arg~min}\limits}

\renewcommand{\tilde}{\widetilde}

\def\f{Fr\'echet }
\def\bco{\iffalse} 
\allowdisplaybreaks

\makeatother

\begin{document}

\title[]{Monotone response for random objects}
\thanks{} 

\author[D. Kurisu]{Daisuke Kurisu}
\author[Y. Okamoto]{Yuta Okamoto}
\author[T. Otsu]{Taisuke Otsu}

\date{First version: \today}

\address[D. Kurisu]{Faculty of Economics, The University of Tokyo 7-3-1, Hongo, Bunkyo-ku, Tokyo 113-0033, Japan.}
\email{d.kurisu@e.u-tokyo.ac.jp}

\address[Y. Okamoto]{Graduate School of Economics, Hitotsubashi University, 2-1 Naka, Kunitachi, Tokyo 186-8601, Japan.}
\email{yuta.okamoto@r.hit-u.ac.jp}

\address[T. Otsu]{Department of Economics, London School of Economics, Houghton Street, London, WC2A 2AE, UK.
}
\email{t.otsu@lse.ac.uk}

\begin{abstract}
Monotone treatment response (MTR), monotone treatment selection (MTS), and monotone instrumental variable (MIV) assumptions are widely used to partially identify counterfactual mean outcomes, but existing analyses have focused almost exclusively on scalar outcomes. We develop a unified framework for partial identification with outcomes that take values in a general metric space under these monotonicity restrictions by embedding the metric space into an $L^2$ space and imposing coordinatewise monotonicity on the embedded functions. The proposed framework yields valid identified sets for Fr\'echet means in a broad class of random-object spaces and further delivers sharp identification results for distributional outcomes under the Wasserstein metric, interval-valued outcomes represented by support functions, and compositional outcomes under the Aitchison metric. We also establish a support-free characterization of the identified set under the joint MTR--MTS assumption. Numerical and empirical illustrations based on Job Corps earnings data and periodontal health distributions from the National Health and Nutrition Examination Survey demonstrate the empirical usefulness of the proposed framework.
\end{abstract}

\maketitle

\section{Introduction} \label{sec:intro}

Monotonicity assumptions such as monotone treatment response (MTR), monotone treatment selection (MTS), and monotone instrumental variables (MIV) have become fundamental assumptions in partial identification analysis. By imposing economically meaningful shape restrictions on counterfactual outcomes, these assumptions substantially tighten identified sets while remaining considerably weaker than conventional point-identifying assumptions. Since the seminal contributions of \cite{Manski:1997} and \cite{Manski_Pepper:2000}, they have been widely applied to study treatment effects in settings where point identification is unattainable.

Existing studies, however, have almost exclusively focused on scalar outcomes. In many empirical applications, outcomes naturally take the form of random objects, namely random elements that take values in a general metric space rather than real-valued variables. Examples include distribution-valued outcomes, such as wage distributions and individual-level distributions of sleep duration; interval-valued outcomes arising from coarsened or censored observations; compositional outcomes, such as firms' allocations of inputs across routine labor, skilled labor, and capital, or households' time allocations across market work, household production, and leisure; and other complex data objects such as covariance matrices that cannot be adequately summarized by a single scalar. In such settings, the conventional MTR, MTS, and MIV assumptions are no longer directly applicable. More fundamentally, there is currently no general framework for extending these monotonicity assumptions to random-object outcomes. Consequently, the identifying implications of monotonicity restrictions for random-object outcomes remain largely unexplored.

This paper develops a unified framework for partial identification with random-object outcomes under the MTR, MTS, and MIV assumptions. Building on an embedding representation of random objects into an $L^2$ space, we formulate monotonicity restrictions coordinatewise on the embedded functions, yielding a general notion of monotonicity applicable to arbitrary metric-space-valued outcomes. We show through a range of examples that appropriate choices of embedding allow the resulting restrictions to capture economically meaningful and intuitively interpretable notions of monotonicity for random objects. Within this framework, we derive tractable valid identified sets for the corresponding Fr\'echet means. We illustrate the framework through two leading applications: distributional outcomes under the Wasserstein metric and interval-valued outcomes represented by support functions.

The proposed framework has several distinctive features. First, it applies to a broad class of random objects through a common embedding representation, allowing MTR, MTS, and MIV restrictions to be formulated in a unified manner. Second, we establish sharp identification results for three economically important classes of random objects: compositional outcomes, distributional outcomes, and interval-valued outcomes, even though sharpness is not guaranteed for the general identified sets. Finally, we illustrate the empirical content of the proposed framework through applications to Job Corps earnings data and periodontal health data from the National Health and Nutrition Examination Survey (NHANES).

Our paper contributes to the literature on partial identification under monotonicity restrictions. Since the seminal contributions of \cite{Manski:1997} and \cite{Manski_Pepper:2000}, monotonicity restrictions including MTR, MTS, and MIV have been extended and applied in a variety of settings. \cite{Molinari:2010} studies partial identification with missing treatment status under MTR and MTS. \cite{Okumura_Usui:2014} investigates the identifying power of concavity restrictions combined with MTR and MTS. \cite{Kim:2018} replaces the MTR assumption with a Lipschitz-type continuity restriction and studies its identifying implications together with MTS. \cite{Jun_Lee:2023} applies the MTR assumption to identify persuasion effects. Despite these developments, the existing literature has focused almost exclusively on scalar outcomes. Our paper extends the MTR, MTS, and MIV framework to random-object outcomes. This paper is also related to the literature on random-object data and metric statistics (see \citet{Dubey:2024} for a survey). Recent work has developed statistical methods for random objects through Fr\'echet means and embedding representations (e.g., \citealp{PeMu19, Dubey:2020, BhMu23}). We contribute to this literature by incorporating partial-identification arguments based on MTR, MTS, and MIV into this framework. To the best of our knowledge, this is the first unified framework for monotonicity restrictions for random-object outcomes.

The remainder of the paper is organized as follows. Section~\ref{sec:prelim} introduces the embedding representation for random objects and reviews the Fr\'echet mean under isometric embeddings. Section~\ref{sec:gen} develops the general identification framework under the MTR, MTS, and MIV assumptions. Sections~\ref{sec:dist} and \ref{sec:interval} establish sharp identification results for distributional and interval-valued outcomes, respectively. Section~\ref{sec:numerical} presents two empirical illustrations of the proposed framework using periodontal health distributions from NHANES and interval-valued Job Corps earnings data. Section~\ref{sec:conclusion} concludes.

\section{Preliminaries} \label{sec:prelim}

\subsection{MTR, MTS, and MIV for scalar outcomes}

We briefly review the scalar-outcome framework of \citet{Manski:1997} and \citet{Manski_Pepper:2000}, which serves as the benchmark for our extension to random-object outcomes. We begin by introducing the concepts of the MTR, MTS, as well as MIV in the standard scalar-outcome setting.

Let $Y_i(\cdot): \mathcal{T}\to\mathcal{Y}$ be a response function mapping the mutually exclusive and exhaustive treatment $t \in \mathcal{T}\subseteq\mathbb{R}$ into outcome $Y_i(t)$. It is typically assumed that $\mathcal{Y}\subseteq \mathbb{R}$, while it is extended to non-Euclidean outcomes in the subsequent sections. Person $i$ has a realized treatment $D_i \in \mathcal{T}$ and a realized outcome $Y_i = Y_i(D_i)$. We are interested in identifying the counterfactual mean $\E{Y_i(t)}$. In a schooling-wage context, this quantity may represent the average counterfactual wage that a worker would earn if they had graduated from high school (i.e., $\E{Y_i(12)}$).

\cite{Manski:1997} and \cite{Manski_Pepper:2000} introduce the following MTR and MTS assumptions:
\begin{assumption}[MTR]\label{assumption:MTR}
    Let $\mathcal{T}$ be an ordered set. 
    Then $t_2 \ge t_1 \Rightarrow Y_i(t_2) \ge Y_i(t_1)$ almost surely.
\end{assumption}

\begin{assumption}[MTS]\label{assumption:MTS}
    Let $\mathcal{T}$ be an ordered set. 
    For each $t \in \mathcal{T}$, $t_2 \ge t_1 \Rightarrow \E{Y_i(t) \mid  D_i=t_2} \ge \E{Y_i(t) \mid  D_i=t_1}$.
\end{assumption}

MTR assumes the monotonicity of the response function $Y_i(\cdot)$. In words, it assumes a weakly monotone relationship between treatment intensity and outcomes---for example, production with respect to inputs or wages with respect to years of schooling \citep{Manski_Pepper:2000}.

While MTR imposes monotonicity of potential outcomes with respect to hypothetical treatments, MTS imposes monotonicity of average potential outcomes with respect to the selected treatment. In the schooling context, MTS assumes that individuals who choose higher levels of schooling have weakly higher mean wage functions than those who choose lower levels of schooling, even when both groups are exogenously assigned the same level of schooling \citep{Manski_Pepper:2000}.

Let $Z_i\in \mathcal{Z}\subseteq\mathbb{R}$ denote the instrumental variable. MTS can be viewed as a special case of the MIV assumption obtained by taking $Z_i=D_i$.
\begin{assumption}[MIV]\label{assumption:MIV}
    Let $\mathcal{Z}$ be an ordered set. $Z_i$ is a monotone instrumental variable in the sense that for each $t\in\mathcal{T}$ and all $(z_1, z_2) \in \mathcal{Z}\times\mathcal{Z}$ such that $z_2\ge z_1$, $\E{Y_i(t) \mid Z_i=z_2} \ge \E{Y_i(t) \mid  Z_i=z_1}$.
\end{assumption}

The MIV assumption extends the monotonicity idea underlying MTS to an observed covariate rather than the treatment itself. Specifically, it assumes that, for any fixed treatment level, the conditional mean of the potential outcome is weakly monotone in the covariate. For instance, in the schooling-wage context, $Z$ may be a proxy for ability \citep{Manski_Pepper:2000}, such as an IQ test score. In this case, MIV assumes that the mean potential outcome at any given level of schooling weakly increases with the IQ score.

\subsection{Random objects and \f mean}

Let $Y_i(\cdot): \mathcal{T}\to\mathcal{Y}$ denote the response function, where $\mathcal{Y}$ is equipped with a metric $d_\mathcal{Y}$. The outcome space $\mathcal{Y}$ need not be a Euclidean space, and $d_\mathcal{Y}$ need not be the Euclidean metric. Leading examples in econometric analysis include distributional outcomes and interval-valued outcomes, neither of which generally takes values in a Euclidean space. For example, \cite{Blundell_etal:2007} studies wage distributions, while \cite{kuri:26} analyzes distributions of individual sleep duration. Interval-valued outcomes arise when data are recorded only in intervals, such as income reported in brackets in surveys \citep{Manski_Tamer:2002}.

Our goal is to extend the monotonicity concepts to such general outcomes, referred to as random objects, and to study their (partial) identification. To this end, we first define an appropriate notion of a ``mean'' for random objects. Because the spaces in which random objects take values generally lack a linear structure, the conventional Euclidean mean is not an intrinsic measure of location. Moreover, even when an expectation can be defined by embedding the object space into a Hilbert space, the resulting quantity need not respect the geometry of the original space. We therefore employ the Fr\'echet mean as the population summary throughout the paper.

Let $\mathbb{E}_{\oplus}[Y_i(t)] \in \mathcal{Y}$ denote the Fr\'echet mean of $Y_i(t)$, defined as
\[
\mathbb{E}_{\oplus}[Y_i(t)]=\argmin_{\omega\in\mathcal Y}\mathbb{E} [d_\mathcal Y^2\!\left(Y_i(t),\omega\right)].
\]
When the minimizer is unique, $\mathbb{E}_{\oplus}[Y_i(t)]$ provides an intrinsic notion of the center of the distribution determined solely by the metric structure of $\mathcal{Y}$. Similarly, the conditional Fr\'echet mean is defined as $\mathbb{E}_\oplus[Y_i(t) \mid \cdot\,]=\mathrm{arg\,min}_{y \in \mathcal{Y}}\mathbb{E}[{d_{\mathcal{Y}}^2(y, Y_i(t))\mid \cdot\,}]$.

To develop our identification results, we adopt the following standard embedding framework from the metric statistics literature. Suppose that there exists an injective map $\Psi:\mathcal Y\rightarrow\mathcal H$, where $\mathcal H$ is a separable Hilbert space, such that
\[
d_{\mathcal Y}(x,y) = d_{\mathcal H}(\Psi(x),\Psi(y))
\]
for all $x,y\in\mathcal Y$, and $\Psi(\mathcal Y)$ is a closed and convex subset of $\mathcal{H}$. Such embeddings are available for many important random objects, including distributional outcomes \citep{Bigot_etal:2017}, compositional outcomes \citep{Aitchison:1982}, and random sets \citep{kuri:25}. Under this framework, the Fr\'echet mean admits the representation
\[
\mathbb{E}_\oplus\!\left[Y_i(t)\mid\cdot\,\right] = \Psi^{-1}\!\left( \mathbb{E}\!\left[\Psi(Y_i(t))\mid\cdot\,\right] \right).
\]
This representation is the key to our identification strategy. It enables us to formulate the identification problem in terms of expectations in the embedded Hilbert space while preserving the geometry of the original object space. More specifically, once an identified set for $\mathbb{E}\!\left[\Psi(Y_i(t))\right]$ is constructed, the corresponding identified set for the Fr\'echet mean $\mathbb{E}_{\oplus}[Y_i(t)]$ is obtained by pulling it back through $\Psi^{-1}$. The next section exploits this representation to extend MTR, MTS, and MIV to random-object outcomes.

\section{General framework} \label{sec:gen}

In this section, we develop a unified framework for MTR, MTS, and MIV with random-object outcomes. Building on the embedding representation introduced in the previous section, we formulate monotonicity restrictions on the embedded random variables and derive identified sets for the corresponding Fr\'echet means.

To facilitate a coordinatewise formulation of monotonicity restrictions, we specialize the general embedding framework introduced in Section~\ref{sec:prelim} to embeddings into an $L^2$ space and impose the following additional regularity conditions.

\begin{assumption}\label{assumption:metric}
There exist a $\sigma$-finite measure space $(\mathcal K,\mathcal A,\nu)$ and an injective map $\Psi:\mathcal Y\to L^2(\mathcal K,\mathcal A,\nu)$ such that $d_{\mathcal Y}(y_1,y_2) = \|\Psi(y_1)-\Psi(y_2)\|_{L^2(\nu)}$. Furthermore, $\Psi(\mathcal Y)$ is a closed convex subset of $L^2(\mathcal K,\mathcal A,\nu)$, and $\mathbb{E}[\|\Psi(Y_i(t))\|_{L^2(\nu)}^2]<\infty$ for every $t\in\mathcal T$.
\end{assumption}

Assumption~\ref{assumption:metric} provides a common representation for a broad class of random objects. Ordinary Euclidean outcomes correspond to $\mathcal K=\{1,\ldots,d\}$ equipped with the counting measure, while distributional outcomes correspond to $L^2((0,1))$, where $\Psi$ is the quantile function. Other examples include interval-valued outcomes, compositional outcomes, and covariance matrices, which are discussed later.

In this representation, monotonicity restrictions can be imposed pointwise on the embedded functions $\Psi(Y_i(t))(k)$ for $k\in\mathcal K$. This coordinatewise formulation is the key to extending the scalar MTR, MTS, and MIV restrictions to random-object outcomes.

We first extend the MTR assumption to random-object outcomes by imposing monotonicity pointwise on the embedded coordinates. Write $X_i^t(k)=\Psi(Y_i(t))(k)$ for $k\in\mathcal K$.
\begin{assumption}[MTR for random objects]\label{assumption:MTR-obj}
    Let $\mathcal{T}$ be an ordered set. There exists a partition $(\mathcal{K}_+, \mathcal{K}_-,\mathcal{K}_0)$ of $\mathcal{K}$ such that, for $t_1,t_2 \in \mathcal{T}$,
    \begin{align*}
        t_2 \ge t_1 \Rightarrow 
        \begin{cases}
            X_i^{t_2}(k) \ge X_i^{t_1}(k) & \text{for a.e. } k \in \mathcal{K}_+,\\
            X_i^{t_2}(k) \le X_i^{t_1}(k) & \text{for a.e. } k \in \mathcal{K}_-,
        \end{cases}
    \end{align*}
    almost surely, and no monotonicity relation is assumed for $k\in\mathcal{K}_0$.
\end{assumption}

This formulation reduces to the usual scalar MTR when $\mathcal K$ is a singleton and to componentwise MTR for multivariate outcomes.

Assumption~\ref{assumption:MTR-obj} requires monotonicity to hold pointwise on the embedded coordinates. The partition $(\mathcal K_+,\mathcal K_-,\mathcal K_0)$ allows different coordinates of the random object to respond monotonically in different directions, while leaving unrestricted coordinates unconstrained. When $\mathcal K_+=\mathcal K$, the assumption corresponds to a globally increasing response, whereas $\mathcal K_-=\mathcal K$ yields a globally decreasing response.

MTR for distributional and interval outcomes, which are among the most common applications, will be treated separately in later sections. Here, we provide some other examples.
\begin{example}[Compositional outcomes]\label{ex:comp}
    Consider the space for compositional data with positive components
    \begin{align*}
        \mathcal Y = \left\{y\in\mathbb{R}^d : y_j > 0,\forall j\in\{1,\ldots,d\}\,\text{ and } \, \sum_{j=1}^d y_j = 1\right\}.
    \end{align*}
    A common metric for compositional data is the Aitchison metric \citep{Aitchison:1982}:
    \begin{align*}
        d_A(x,y) = \|\Lambda(x) - \Lambda(y)\|_{\mathbb R^d},
    \end{align*}
    where $\|\cdot\|_{\mathbb{R}^d}$ is the standard Euclidean norm on $\mathbb{R}^d$ and 
    \begin{align*}
        \Lambda(x) = \left(\log {\frac{x_1}  {g(x)}},\dots, \log {\frac{x_d} {g(x)}}\right)^\top,\quad g(x) = \left(\prod_{j=1}^d x_j\right)^{1/d}.
    \end{align*}    
    Although $\Lambda$ is an isometric embedding into $\mathbb{R}^d$, defining monotonicity through this representation may be less intuitive in some applications because each coordinate $x_k/g(x)$ depends on the geometric mean of all components of $x$. We therefore consider the following alternative embedding. Take a matrix $V\in \mathbb{R}^{d \times (d-1)}$ satisfying $ V^\top V = I_{d-1}$ and $V^\top 1_d = 0$, where $1_d = (1,\ldots, 1)$. We define $\Psi(x) \coloneqq V^\top \Lambda(x)$. We then have
    \begin{align*}
        d_A(x,y) = \|\Lambda(x) - \Lambda(y)\|_{\mathbb R^d}
        = \|V^\top \Lambda(x) - V^\top \Lambda(y)\|_{\mathbb R^{d-1}}
        = \|\Psi(x) - \Psi(y)\|_{\mathbb R^{d-1}},
    \end{align*}
    and $\Psi(\mathcal Y) = \mathbb R^{d-1}$. Hence $\Psi$ is an isometry and $\Psi(\mathcal Y)$ is closed convex. This $\Psi$ is an isometric log-ratio transform associated with the orthonormal basis $V$; see \cite{Egozcue03} for a general treatment.

    We now provide a possible MTR restriction.
    To illustrate, consider the three-part compositional outcome $Y=(y_1,y_2,y_3)$. Take $V=[v_1\ v_2]$ with $v_1 = (1/\sqrt{2}, -1/\sqrt{2}, 0)^\top$ and $v_2 = (1/\sqrt{6}, 1/\sqrt{6}, -2/\sqrt{6})^\top$. Then we have 
    \begin{align*}
        \Psi(Y) = \begin{pmatrix}
            \Psi(Y)(1)\\
            \Psi(Y)(2)
        \end{pmatrix}
        =
        \begin{pmatrix}
            \dfrac{1}{\sqrt{2}}\log \left(\dfrac{y_1}{y_2}\right)\\
            \sqrt{\dfrac{2}{3}}\log \left(\dfrac{\sqrt{y_1 y_2}}{y_3}\right)
        \end{pmatrix}.
    \end{align*}
    To fix ideas, suppose that $y_1$, $y_2$, and $y_3$ denote a firm's expenditure shares on skilled labor, routine labor, and capital, respectively. Let the treatment $t$ measure the introduction or intensity of AI and automation technologies. MTR may impose $\Psi(Y_i(t_2))(1)\geq \Psi(Y_i(t_1))(1)$ and $\Psi(Y_i(t_2))(2)\leq \Psi(Y_i(t_1))(2)$ whenever $t_2\geq t_1$. The first restriction can be interpreted as assuming that greater adoption of AI and automation technologies weakly increases the use of skilled labor relative to routine labor. The second restriction assumes that greater adoption weakly decreases the (geometric) mean of the two labor components relative to capital.\hfill\qed
\end{example}

\begin{example}[Covariance matrices]
    Let $\mathrm{Sym}_m^+$ denote the space of $m \times m$ symmetric positive-definite matrices. Let $\lambda_{\mathrm{min}}(A)$ be the minimum eigenvalue of $A \in \mathrm{Sym}_m^+$. For any $\varepsilon > 0$, define the space $\mathrm{Sym}_{m,\varepsilon}^+=\{A \in \mathrm{Sym}_m^+: \lambda_{\mathrm{min}}(A)\ge \varepsilon\}$. Suppose the outcome $Y$ takes values in $\mathcal Y = \mathrm{Sym}_{m,\varepsilon}^+$.
    
    A commonly employed metric is the Frobenius metric $d_F(A,B) = [\mathrm{tr}\{(A-B)^\prime(A-B)\}]^{1/2} = \|\Lambda(A) - \Lambda(B)\|$ where $\Lambda(A)$ is the vectorization of the matrix $A$, that is, $\Lambda(A) = \mathrm{vec}(A)$ and $\|\cdot\|$ is the Euclidean norm on $\mathbb{R}^{m^2}$. However, coordinatewise monotonicity on $\operatorname{vec}(A)$ may be difficult to motivate in some applications. In particular, the resulting ordering does not generally coincide with conventional matrix orderings.

    Let $\Psi(A)(u) = u^\top A u$ for $u\in\mathbb S^{m-1}$. Then $\Psi(\mathcal Y)$ is a closed convex subset of $L^2(\mathbb S^{m-1})$.
    We introduce the following metric
    \begin{align*}
        d_S(A,B)^2 = 
        \|\Psi(A) - \Psi(B)\|_{L^2(\nu)}^2 = 
        \int_{\mathbb S^{m-1}} \left\{u^\top (A-B) u\right\}^2\,d\nu(u),
    \end{align*}
    where $\nu$ denotes the uniform probability measure on $\mathbb S^{m-1}$. It can be shown that the \f means induced by the metrics $d_F$ and $d_S$ coincide, so the target parameter is invariant to the choice between these two metrics. Under this embedding, a possible MTR restriction is $\Psi(Y_i(t_2))(u) \ge \Psi(Y_i(t_1))(u) \Leftrightarrow u^\top Y_i(t_2) u \ge u^\top Y_i(t_1) u$ for all $u \in \mathbb S^{m-1}$. This is equivalent to the conventional Loewner ordering $Y_i(t_2)\succeq Y_i(t_1)$. Intuitively speaking, this MTR means that outcomes under treatment $t_2$ are weakly more dispersed than those under treatment $t_1$.\hfill\qed
\end{example}

As illustrated in the examples above, the representation of a random object by an embedding map $\Psi$ and the choice of metric $d_{\mathcal Y}$ need not be unique. The choice of $\Psi$ is substantively relevant in the present framework because the coordinatewise MTR restriction (and MTS and MIV introduced below as well) is imposed on $\Psi(Y)$. Thus, even when two embeddings are isometric with respect to the same metric, they may generate different notions of coordinatewise monotonicity. An embedding should therefore be chosen so that its coordinates represent economically meaningful features of the random object and the imposed directions of monotonicity can be justified in the application. For example, different choices of the orthonormal basis $V$ in the compositional-data example correspond to monotonicity restrictions on different log-contrasts.

Regarding the choice of metric, many classes of random objects are equipped with a standard metric that reflects their commonly accepted geometry. Such a metric provides a natural benchmark for defining the target parameter. An alternative metric may also be considered when it facilitates the formulation of economically interpretable monotonicity restrictions, provided that the resulting target parameter coincides with that defined under the benchmark metric.

We next extend the MTS assumption to random-object outcomes by imposing monotonicity on the conditional expectations of the embedded coordinates. As in Assumption~\ref{assumption:MTR-obj}, the monotonicity direction may differ across coordinates through the partition $(\mathcal K_+,\mathcal K_-,\mathcal K_0)$.

\begin{assumption}[MTS for random objects]\label{assumption:MTS-obj}
    Let $\mathcal{T}$ be an ordered set. 
    For each $t\in\mathcal T$ and all $t_1,t_2\in\mathcal T$,
    \begin{align*}
        t_2 \ge t_1 \Rightarrow 
        \begin{cases}
            \mathbb{E}[X_i^{t}(k) \mid  D_i=t_2] \ge \mathbb{E}[X_i^{t}(k) \mid  D_i=t_1] & \text{for a.e. } k \in \mathcal{K}_+,\\
            \mathbb{E}[X_i^{t}(k) \mid  D_i=t_2] \le \mathbb{E}[X_i^{t}(k) \mid  D_i=t_1] & \text{for a.e. } k \in \mathcal{K}_-,
        \end{cases}
    \end{align*}
    and no monotonicity relation is assumed for $k\in\mathcal{K}_0$.
\end{assumption}

Assumption~\ref{assumption:MTS-obj} extends the scalar MTS restriction by requiring the conditional Fr\'echet mean, represented through the embedded coordinates, to vary monotonically with the realized treatment. As before, coordinates in $\mathcal K_0$ are left unrestricted, while coordinates in $\mathcal K_+$ and $\mathcal K_-$ are allowed to satisfy monotonicity in opposite directions.

We finally extend the MIV assumption to random-object outcomes by imposing monotonicity on the conditional expectations of the embedded coordinates with respect to the instrumental variable. As before, the monotonicity direction is allowed to differ across coordinates through the partition $(\mathcal K_+,\mathcal K_-,\mathcal K_0)$.

\begin{assumption}[MIV for random objects]\label{assumption:MIV-obj}
    Let $\mathcal{Z}$ be an ordered set. 
    For each $t\in\mathcal{T}$ and all $(z_1, z_2) \in \mathcal{Z}\times\mathcal{Z}$ such that $z_2\ge z_1$, 
    \begin{align*}
            \mathbb{E}[X_i^{t}(k) \mid Z_i=z_2] \ge \mathbb{E}[X_i^{t}(k) \mid  Z_i=z_1] & \quad\text{for a.e. } k \in \mathcal{K}_+,\\
            \mathbb{E}[X_i^{t}(k) \mid Z_i=z_2] \le \mathbb{E}[X_i^{t}(k) \mid  Z_i=z_1] & \quad\text{for a.e. } k \in \mathcal{K}_-,
    \end{align*}
    and no monotonicity relation is assumed for $k\in\mathcal{K}_0$.
\end{assumption}

Assumption~\ref{assumption:MIV-obj} extends the scalar MIV restriction to random-object outcomes by requiring the conditional Fr\'echet mean, represented through the embedded coordinates, to vary monotonically with the instrumental variable. As in the MTR and MTS assumptions, unrestricted coordinates in $\mathcal K_0$ are left unconstrained, while coordinates in $\mathcal K_+$ and $\mathcal K_-$ may satisfy monotonicity in opposite directions. The three assumptions above provide coordinatewise extensions of the scalar MTR, MTS, and MIV restrictions to random-object outcomes. We now show that these assumptions imply tractable identified sets for the corresponding Fr\'echet means.

Appendix~\ref{sec:gen-set} derives sharp identified sets for general random objects under each assumption. These sharp sets, however, may be difficult to characterize or compute in general. We therefore begin with simpler identified sets obtained by combining the coordinatewise sharp bounds of \cite{Manski:1997} and \cite{Manski_Pepper:2000}. As shown in Propositions~\ref{prop:simple-sharp}--\ref{prop:int} below, some of these identified sets are in fact sharp for several important classes of random objects. The following proposition establishes valid identified sets under each of the three monotonicity assumptions.

\begin{proposition}\label{prop:gen-simple}
    Suppose that Assumption \ref{assumption:metric} holds true. Assume that $\mathcal K_+$, $\mathcal K_-$, and $\mathcal K_0$ are measurable and form a partition of $\mathcal K$ up to $\nu$-null sets. Suppose that there exist measurable functions $k \mapsto \underline{\mathcal X}_k$ and $k \mapsto \overline{\mathcal X}_k$ such that, for every $s \in \mathcal T$, $\P{\underline{\mathcal X}_k \le X_i^s(k) \le \overline{\mathcal X}_k} = 1$ holds for $\nu$-almost every $k\in\mathcal K$. We also define $X_i(k) = X_i^{D_i}(k)$. Let $\mathcal T$ and $\mathcal Z$ be discrete, finite sets. Finally, assume $\P{D_i = d}>0$ for every $d \in \mathcal T$.
    \begin{description}
        \item[(i)] Suppose that Assumption \ref{assumption:MTR-obj} holds true and that $\overline{\mathcal X}_k$ and $\underline{\mathcal X}_k$ are known. Then, a valid identified set of $\E{\Psi(Y_i(t))}$ is given by $\Psi(\mathcal Y) \cap \mathcal B_{\mathtt{MTR}}(t)$, where 
        \begin{align*}
            \mathcal B_{\mathtt{MTR}}(t) = 
            \left\{x\in L^2(\mathcal K, \mathcal A, \nu): \ell^{\mathtt{MTR}}_{t,k} \le x(k)\le u^{\mathtt{MTR}}_{t,k} \text{ for a.e. }k\in\mathcal K \right\},
        \end{align*}
        and the coordinate-wise bounds are defined as follows:
        For $k\in\mathcal K_+$,
        \[ 
            \ell^{\mathtt{MTR}}_{t,k} 
            = \mathbb{E}[X_i(k)\mathbf 1\{D_i\le t\}] +\underline{\mathcal X}_k\P{D_i>t},\quad
            u^{\mathtt{MTR}}_{t,k}
            = \mathbb{E}[X_i(k)\mathbf 1\{D_i\ge t\}] +\overline{\mathcal X}_k\P{D_i<t}. 
        \]
        For $k\in\mathcal K_-$, 
        \[
            \ell^{\mathtt{MTR}}_{t,k} 
            = \mathbb{E}[X_i(k)\mathbf 1\{D_i\ge t\}] +\underline{\mathcal X}_k\P{D_i<t},\quad
            u^{\mathtt{MTR}}_{t,k} 
            = \mathbb{E}[X_i(k)\mathbf 1\{D_i\le t\}] +\overline{\mathcal X}_k\P{D_i>t}.
        \]
        For $k\in\mathcal K_0$, 
        \[
            \ell^{\mathtt{MTR}}_{t,k} 
            = \mathbb{E}[X_i(k)\mathbf 1\{D_i=t\}] +\underline{\mathcal X}_k\P{D_i\ne t},\quad 
            u^{\mathtt{MTR}}_{t,k} 
            = \mathbb{E}[X_i(k)\mathbf 1\{D_i=t\}] +\overline{\mathcal X}_k\P{D_i\ne t}. 
        \]

        \item[(ii)] Suppose that Assumption \ref{assumption:MTS-obj} holds true and that $\overline{\mathcal X}_k$ and $\underline{\mathcal X}_k$ are known. Then, a valid identified set of $\E{\Psi(Y_i(t))}$ is given by $\Psi(\mathcal Y) \cap \mathcal B_{\mathtt{MTS}}(t)$, where 
        \begin{align*}
            \mathcal B_{\mathtt{MTS}}(t) = 
            \left\{x\in L^2(\mathcal K, \mathcal A, \nu): \ell^{\mathtt{MTS}}_{t,k} \le x(k)\le u^{\mathtt{MTS}}_{t,k} \text{ for a.e. }k\in\mathcal K \right\},
        \end{align*}
        and the coordinate-wise bounds are defined as follows:
        For $k\in\mathcal K_+$, 
        \[ 
            \ell^{\mathtt{MTS}}_{t,k} 
            = \mathbb{E}[X_i(k)\mid D_i= t]\P{D_i\ge t} +\underline{\mathcal X}_k\P{D_i<t},\quad
            u^{\mathtt{MTS}}_{t,k}
            = \mathbb{E}[X_i(k) \mid D_i= t]\P{D_i\le t} +\overline{\mathcal X}_k\P{D_i>t}. 
        \]
        For $k\in\mathcal K_-$, 
        \[ 
            \ell^{\mathtt{MTS}}_{t,k} 
            = \mathbb{E}[X_i(k)\mid D_i= t]\P{D_i\le t} +\underline{\mathcal X}_k\P{D_i>t},\quad
            u^{\mathtt{MTS}}_{t,k}
            = \mathbb{E}[X_i(k) \mid D_i= t]\P{D_i\ge t} +\overline{\mathcal X}_k\P{D_i<t}. 
        \]
        For $k\in\mathcal K_0$, $\ell^{\mathtt{MTS}}_{t,k} = \ell^{\mathtt{MTR}}_{t,k}$ and $u^{\mathtt{MTS}}_{t,k} = u^{\mathtt{MTR}}_{t,k}$.

        \item[(iii)] Suppose that Assumption \ref{assumption:MIV-obj} holds true and that $\overline{\mathcal X}_k$ and $\underline{\mathcal X}_k$ are known. 
        Suppose also that $\P{D_i=d,Z_i=z}>0$ for every $(d,z) \in \mathcal T \times \mathcal Z$.
        Then, a valid identified set of $\E{\Psi(Y_i(t))}$ is given by $\Psi(\mathcal Y) \cap \mathcal B_{\mathtt{MIV}}(t)$, where 
        \begin{align*}
            \mathcal B_{\mathtt{MIV}}(t) = 
            \left\{x\in L^2(\mathcal K, \mathcal A, \nu): \ell^{\mathtt{MIV}}_{t,k} \le x(k)\le u^{\mathtt{MIV}}_{t,k} \text{ for a.e. }k\in\mathcal K \right\},
        \end{align*}
        and the coordinate-wise bounds are defined as follows:
        For $k\in\mathcal K_+$, 
        \begin{align*} 
            \ell^{\mathtt{MIV}}_{t,k}
            &=\sum_{z\in\mathcal Z}\P{Z_i=z}\sup_{z^\prime\le z}\left(\mathbb{E}[X_i(k)\mathbf 1\{D_i=t\}\mid Z_i=z^\prime] + \underline{\mathcal X}_k \P{D_i\ne t\mid Z_i=z^\prime}\right),\\
            u^{\mathtt{MIV}}_{t,k}
            &=\sum_{z\in\mathcal Z}\P{Z_i=z}\inf_{z^\prime\ge z}\left(\mathbb{E}[X_i(k)\mathbf 1\{D_i=t\}\mid Z_i=z^\prime] + \overline{\mathcal X}_k \P{D_i\ne t\mid Z_i=z^\prime}\right).
        \end{align*} 
        For $k\in\mathcal K_-$, 
        \begin{align*} 
            \ell^{\mathtt{MIV}}_{t,k}
            &=
            \sum_{z\in\mathcal Z}
            \P{Z_i=z}
            \sup_{z^\prime\ge z}
            \left(\E{X_i(k)\mathbf 1\{D_i=t\}\mid Z_i=z^\prime} + \underline{\mathcal X}_k \P{D_i\ne t\mid Z_i=z^\prime}\right),\\
            u^{\mathtt{MIV}}_{t,k}
            &=
            \sum_{z\in\mathcal Z}
            \P{Z_i=z}
            \inf_{z^\prime\le z}
            \left(\E{X_i(k)\mathbf 1\{D_i=t\}\mid Z_i=z^\prime} + \overline{\mathcal X}_k \P{D_i\ne t\mid Z_i=z^\prime}\right).
        \end{align*} 
        For $k\in\mathcal K_0$, $\ell^{\mathtt{MIV}}_{t,k} = \ell^{\mathtt{MTR}}_{t,k}$ and $u^{\mathtt{MIV}}_{t,k} = u^{\mathtt{MTR}}_{t,k}$.
    \end{description}
    Furthermore, $\Psi^{-1}(\Psi(\mathcal Y) \cap \mathcal B_{\mathtt{M}}(t))$ is also a valid identified set of $\mathbb{E}_{\oplus} [Y_i(t)]$, where $\mathtt{M}\in\{\mathtt{MTR}, \mathtt{MTS}, \mathtt{MIV}\}$.
\end{proposition}

\begin{remark}
Most empirical applications involve treatments and instruments with finite support. Accordingly, throughout the paper we focus on finite sets $\mathcal T$ and $\mathcal Z$, as in \cite{Manski_Pepper:2000}. This finiteness restriction simplifies the exposition and the proofs. Nevertheless, the identification arguments could be extended to continuous treatments and instruments by replacing finite sums with integrals with respect to the corresponding probability measures, subject to appropriate measurability and regularity conditions. Such an extension may not be entirely mechanical for the sharpness result for distributional outcomes in Section~\ref{sec:dist}, however, because its proof explicitly exploits the finiteness of $\mathcal T$.
\end{remark}

When $(\mathcal{Y},d_\mathcal{Y})=(\mathbb{R}, \|\cdot\|_\mathbb{R})$, the coordinatewise bounds in Proposition~\ref{prop:gen-simple} coincide with the sharp scalar bounds of \cite{Manski:1997} and \cite{Manski_Pepper:2000}. In the random-object setting, however, the resulting identified sets are generally not sharp because they ignore cross-coordinate restrictions implied by the geometry of $\Psi(\mathcal Y)$. Sections~\ref{sec:dist} and \ref{sec:interval} show that, in appropriately defined senses, these coordinatewise identified sets are sharp for interval-valued outcomes and, except under MIV, for distribution-valued outcomes.

Another limitation of the bounds in parts (i)–(iii) is that they depend on the coordinatewise support limits $\underline{\mathcal X}_k$ and $\overline{\mathcal X}_k$.  This limitation is shared by the original scalar case of \cite{Manski_Pepper:2000}. As in the scalar case, however, this dependence disappears when MTR and MTS are imposed jointly, provided that $\mathcal K_0 = \varnothing$. The next part of Proposition~\ref{prop:gen-simple} extends the corresponding result of \cite{Manski_Pepper:2000} to random-object outcomes.

\addtocounter{proposition}{-1}
\begin{proposition}[continued]\label{prop:gen-simple-cont}
    \quad
    \begin{description}
        \item[(iv)] Suppose that Assumptions \ref{assumption:MTR-obj} and \ref{assumption:MTS-obj} hold and that $\mathcal K_0=\varnothing$. Then, a valid identified set of $\E{\Psi(Y_i(t))}$ is given by $\Psi(\mathcal Y) \cap \mathcal B_{\mathtt{MTR+MTS}}(t)$, where 
        \begin{align*}
            \mathcal B_{\mathtt{MTR+MTS}}(t)
            =
            \left\{
            x\in L^2(\mathcal K, \mathcal A, \nu):
            \ell^{\mathtt{MTR+MTS}}_{t,k}
            \le x(k)\le
            u^{\mathtt{MTR+MTS}}_{t,k}
            \text{ for a.e. } k\in\mathcal K
            \right\},
        \end{align*}
        and the coordinate-wise bounds are defined as follows:
        For $k\in\mathcal K_+$, 
        \begin{align*}
            \ell^{\mathtt{MTR+MTS}}_{t,k}
            &=
            \sum_{d < t}\E{X_i(k) \mid D_i = d}\P{D_i=d} 
            +
            \E{X_i(k) \mid D_i = t}\P{D_i\ge t},\\
            u^{\mathtt{MTR+MTS}}_{t,k}
            &=
            \sum_{d > t}\E{X_i(k) \mid D_i = d}\P{D_i=d} 
            +
            \E{X_i(k) \mid D_i = t}\P{D_i\le t}.
        \end{align*}
        For $k\in\mathcal K_-$,
        \begin{align*}
            \ell^{\mathtt{MTR+MTS}}_{t,k}
            &=\sum_{d > t}\mathbb{E}[X_i(k) \mid D_i = d]\P{D_i=d} +\mathbb{E}[X_i(k) \mid D_i = t]\P{D_i\le t},\\
            u^{\mathtt{MTR+MTS}}_{t,k}
            &=\sum_{d < t}\E{X_i(k) \mid D_i = d}\P{D_i=d} +\mathbb{E}[X_i(k) \mid D_i = t]\P{D_i\ge t}.
        \end{align*}
        Furthermore, $\Psi^{-1}(\Psi(\mathcal Y) \cap \mathcal B_{\mathtt{MTR+MTS}}(t))$ is a valid identified set of $\mathbb{E}_{\oplus} [Y_i(t)]$.
    \end{description}
\end{proposition}

Proposition~\ref{prop:gen-simple}(i)--(iv) establishes a general framework for partial identification with random-object outcomes under monotonicity restrictions. 

Although the identified set in Proposition \ref{prop:gen-simple} is not sharp in general, it is sharp in several important cases. In the subsequent sections, we separately consider important classes of random objects, including distributional and interval-valued outcomes. Another special case arises when $\Psi(\mathcal Y)=\mathbb R^p$, as summarized in the following proposition.

\begin{proposition}\label{prop:simple-sharp}
In addition to the assumptions corresponding to each part of Proposition~\ref{prop:gen-simple}, %including the additional condition $\mathcal K_0=\varnothing$ for $\mathtt M=\mathtt{MTR+MTS}$, 
suppose that the embedding space is $\mathbb R^p$ for a finite positive integer $p$ and that $\Psi(\mathcal Y)=\mathbb R^p$. Then, for each $\mathtt M\in\{\mathtt{MTR},\mathtt{MTS},\mathtt{MIV}, \mathtt{MTR+MTS}\}$, $\mathcal B_{\mathtt M}(t)$ is the sharp identified set of $\E{\Psi(Y_i(t))}$. Moreover, $\Psi^{-1}(\mathcal B_{\mathtt M}(t))$ is the sharp identified set of $\mathbb E_\oplus[Y_i(t)]$.
\end{proposition}

This proposition implies, for example, that the identified set derived under the general framework is sharp for compositional outcomes with strictly positive components (Example \ref{ex:comp}).

\subsection{Numerical illustration}
To illustrate the relationship between the bounds in the embedding space and the identified set for the \f mean, we provide a brief numerical illustration of the combined MTR--MTS bounds for three-part compositional outcomes.

Let $D_i \in \mathcal T = \{0,1,2\}$ with $\P{D_i = 0} = 1/4$, $\P{D_i = 1} = 1/2$, and $\P{D_i = 2} = 1/4$.
We specify
\begin{align*}
    \Psi(Y_i(t)) = 
    \begin{pmatrix}
        -0.4+0.4D_i+0.2t+\varepsilon_{i1}\\
        \phantom{-}0.7-0.4D_i-0.2t+\varepsilon_{i2}
    \end{pmatrix},
\end{align*}
where $\E{\varepsilon_{i1}\mid D_i}=\E{\varepsilon_{i2}\mid D_i}=0$.
Let $\mathcal K_+ = \{1\}$, $\mathcal K_- = \{2\}$, and $\mathcal K_0 = \varnothing$.
Then, this DGP satisfies the proposed MTR restrictions. In particular, for $t_2\geq t_1$, $\Psi(Y_i(t_2))(1)-\Psi(Y_i(t_1))(1) \ge 0$ and $\Psi(Y_i(t_2))(2)-\Psi(Y_i(t_1))(2) \le 0$.
The DGP also satisfies MTS because $\E{\Psi(Y_i(t))(1)\mid D_i=d}$ is increasing in $d$, while $\E{\Psi(Y_i(t))(2)\mid D_i=d}$ is decreasing in $d$.

%Note $\E{\Psi(Y_i) \mid D_i =d} = (-0.4 + 0.6 d,\ 0.7 - 0.6d)^\top$.
The endpoints of the identified sets characterized in Proposition~\ref{prop:gen-simple-cont}(iv) can be computed as
\begin{align*}
    (\ell^{\mathtt{MTR+MTS}}_{0,1}, u^{\mathtt{MTR+MTS}}_{0,1}, \ell^{\mathtt{MTR+MTS}}_{0,2}, u^{\mathtt{MTR+MTS}}_{0,2}) &= (-0.40, 0.20, 0.10, 0.70 ),\\
    (\ell^{\mathtt{MTR+MTS}}_{1,1}, u^{\mathtt{MTR+MTS}}_{1,1}, \ell^{\mathtt{MTR+MTS}}_{1,2}, u^{\mathtt{MTR+MTS}}_{1,2}) &= (0.05, 0.35, -0.05, 0.25),\\
    (\ell^{\mathtt{MTR+MTS}}_{2,1}, u^{\mathtt{MTR+MTS}}_{2,1}, \ell^{\mathtt{MTR+MTS}}_{2,2}, u^{\mathtt{MTR+MTS}}_{2,2}) &= (0.20, 0.80, -0.50, 0.10).
\end{align*}
The corresponding identified regions in the embedding space are shown in
Figure~\ref{fig:ilr}.
These regions characterize the subsets of $\mathbb R^2$ that are consistent with the observed data and the imposed assumptions.

\begin{figure}[t]
    \centering
    \begin{subfigure}[b]{0.49\textwidth}
        \centering
        \includegraphics[width=\linewidth]{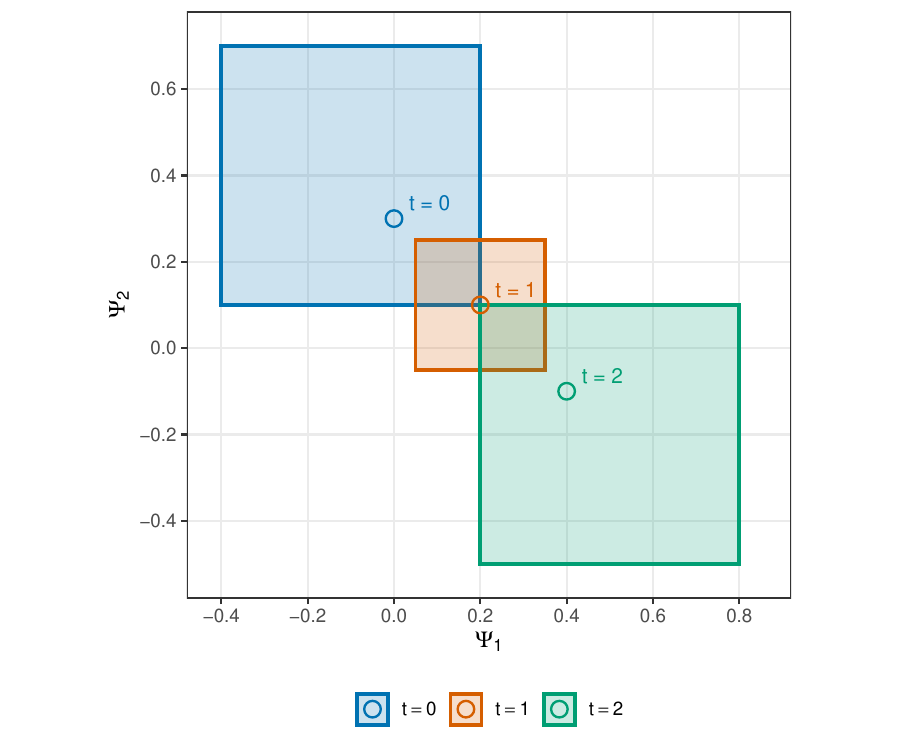}
        \caption{Identified regions in the embedding space}
        \label{fig:ilr}
    \end{subfigure}
    \hfill
    \begin{subfigure}[b]{0.49\textwidth}
        \centering
        \includegraphics[width=\linewidth]{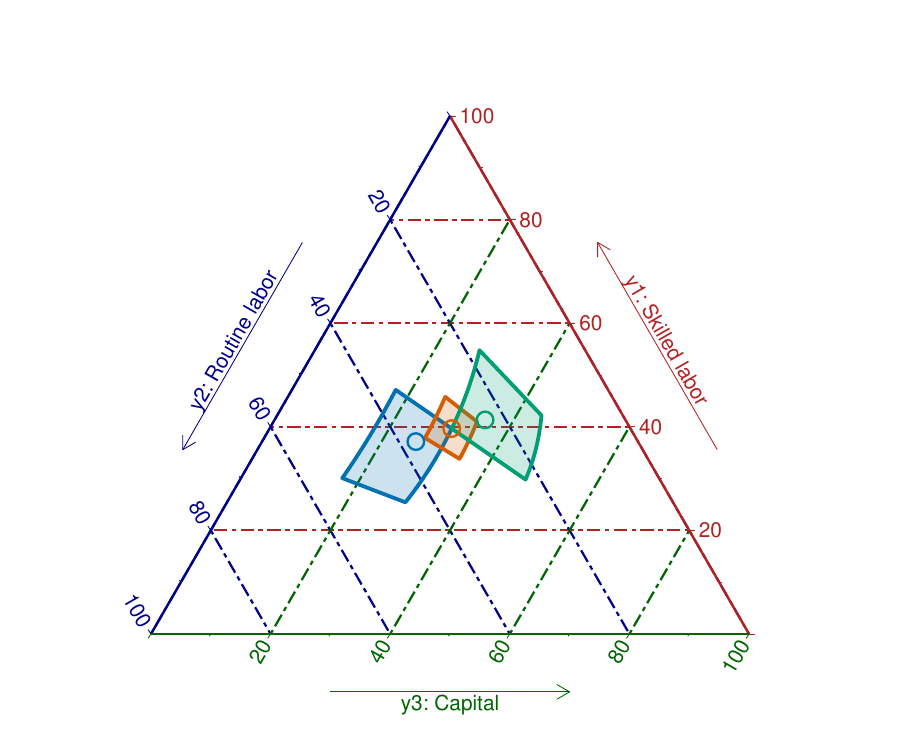}
        \caption{Identified sets for the \f mean in the simplex}
        \label{fig:tern}
    \end{subfigure}
    \caption{Numerical illustration for compositional outcomes}
    \label{fig:comp}

    \begin{flushleft}
    \footnotesize
    \renewcommand{\baselineskip}{11pt}
    \textbf{Note:} The identified sets of $\E{\Psi(Y_i(t))}$ and $\mathbb{E}_\oplus[Y_i(t)]$ are shown for $t\in\{0,1,2\}$. 
    From left to right, the boxed regions correspond to the identified sets for $t=0,1,2$, respectively.
    The hollow circles indicate the corresponding true mean values.
    The component labels are drawn from Example~\ref{ex:comp} for illustrative purposes.
    \end{flushleft}
\end{figure}

Taking the inverse images of these regions yields the identified sets for the \f mean. Figure~\ref{fig:tern} presents the resulting sets in a ternary plot.
By Proposition~\ref{prop:simple-sharp}, these are the sharp identified sets for $\mathbb{E}_\oplus[Y_i(t)]$ for $t\in\{0,1,2\}$.

\section{Distributional outcomes} \label{sec:dist}

In this section, we apply the general framework developed in Section~\ref{sec:gen} to distributional outcomes. Distributional data arise naturally in many econometric applications, including analyses of wage distributions \citep{Blundell_etal:2007}, job-tenure distributions \citep{vanDijcke_etal:2024}, and distributions of sleep duration \citep{kuri:26}. Under the 2-Wasserstein metric, the embedding introduced in Section~\ref{sec:prelim} is given by the quantile function, which allows the general identification results to be strengthened. In particular, we show that the coordinatewise bounds in Proposition~\ref{prop:gen-simple} yield uniformly sharp lower and upper bounds for distributional outcomes under MTR, MTS, and joint MTR--MTS, while the MIV case requires a different construction.

Let $\mathcal Y$ denote the space of probability measures on the real line with finite second moments, equipped with the 2-Wasserstein metric
\[
d_{\mathcal Y}(\mu_1,\mu_2) = \left[ \int_0^1 \{Q_{\mu_1}(q)-Q_{\mu_2}(q)\}^2\,dq \right]^{1/2},
\]
where $Q_\mu$ denotes the quantile function of $\mu$, i.e., $Q_\mu(q) = \inf \{x\in\mathbb R: F_\mu(x) \ge q\}$, where $F_\mu$ is the distribution function. Let $\Psi(\mu)=Q_\mu$. Then
\[
d_{\mathcal Y}(\mu_1,\mu_2) = \|\Psi(\mu_1)-\Psi(\mu_2)\|_{L^2},
\]
so that $(\mathcal Y,d_{\mathcal Y})$ is isometrically embedded into $L^2((0,1))$. Moreover, $\Psi(\mathcal Y)$ is a closed convex subset of $L^2((0,1))$ \citep[Proposition 2.1]{Bigot_etal:2017}. Consequently, the general framework developed in Section~\ref{sec:gen} applies directly with the quantile function as the embedding map.

For distributional outcomes, the coordinatewise monotonicity restrictions introduced in Section~\ref{sec:gen} reduce to monotonicity restrictions on the quantile functions. This observation motivates imposing first-order stochastic dominance (FOSD) as the distributional analogue of MTR, as in \cite{Blundell_etal:2007}. This yields the following specialization of the MTR assumption. For notational convenience, write $Q_i(q)=Q_{Y_i}(q)$ for the quantile function of the observed outcome $Y_i=Y_i(D_i)$.

\begin{assumption}[MTR for distributional outcomes]\label{assumption:D-MTR}
    Let $\mathcal T$ be an ordered set. For all $t_1,t_2\in\mathcal T$,
    \begin{align*}
        t_2\ge t_1
        \Longrightarrow
        Q_{Y_i(t_2)}(q)\ge Q_{Y_i(t_1)}(q)
        \quad\text{for all }q\in(0,1),
    \end{align*}
    almost surely.
\end{assumption}

Assumption~\ref{assumption:D-MTR} requires $Y_i(t_2)$ to first-order stochastically dominate $Y_i(t_1)$ whenever $t_2\ge t_1$. Since FOSD is one of the most widely used partial orderings for distributions in economics (see, e.g., \citealp{Whang:2019}), it provides a natural distributional analogue of the scalar MTR assumption.

The MTS and MIV assumptions can be extended analogously by imposing monotonicity on the expected quantile functions. This corresponds to an expected version of first-order stochastic dominance and yields the following assumptions.

\begin{assumption}[MTS for distributional outcomes]\label{assumption:D-MTS}
    Let $\mathcal{T}$ be an ordered set. 
    For each $t\in\mathcal T$ and all $t_1,t_2\in\mathcal T$,
    \begin{align*}
        t_2 \ge t_1 \Longrightarrow \mathbb{E}[Q_{Y_i(t)}(q) | D_i=t_2] \ge \mathbb{E}[Q_{Y_i(t)}(q) | D_i=t_1]\quad \mbox{ for all } q\in(0,1).
    \end{align*}
\end{assumption}

Assumption~\ref{assumption:D-MTS} requires the conditional mean quantile function to vary monotonically with the realized treatment. This is the distributional analogue of the scalar MTS assumption. We similarly extend the MIV assumption by requiring monotonicity of the conditional mean quantile function with respect to the instrumental variable.

\begin{assumption}[MIV for distributional outcomes]\label{assumption:D-MIV}
    Let $\mathcal{Z}$ be an ordered set. 
    For each $t\in\mathcal{T}$ and all $(z_1, z_2) \in \mathcal{Z}\times\mathcal{Z}$ such that $z_2\ge z_1$, 
    \begin{align*}
        \mathbb{E}[Q_{Y_i(t)}(q) \mid Z_i = z_2] \ge \mathbb{E}[Q_{Y_i(t)}(q) \mid Z_i = z_1]\quad \mbox{for all }  q\in(0,1).
    \end{align*}
\end{assumption}

The target objects in this section are represented as quantile functions. We therefore introduce a practically useful notion of sharpness that requires the extremal functions themselves to belong to the identified set. The following definition formalizes this notion. See also \citet[Definition 2]{Firpo_Ridder:2019} for the related concept of uniform sharpness for cumulative distribution functions.

\begin{definition}[Uniformly sharp bounds]
    Let $\mathcal S$ be the sharp identified set for a functional parameter over an index set $\mathcal K$. Then $g(\cdot)$ is called a uniformly sharp lower bound if (i) $f(k)\ge g(k)$ for every $f\in\mathcal S$ and every $k\in\mathcal K$, and (ii) $g\in\mathcal S$. 
    Analogously, $g(\cdot)$ is called a uniformly sharp upper bound if (i) $f(k)\le g(k)$ for every $f\in\mathcal S$ and every $k\in\mathcal K$, and (ii) $g\in\mathcal S$.
\end{definition}

The following proposition establishes uniformly sharp bounds under the distributional MTR, MTS, and MIV assumptions. Throughout, let
\begin{align*}
    \mathcal Q_{[K_0, K_1]} = \left\{
    Q_\mu : \mu \in \mathcal Y \text{ and } Q_\mu(q) \in [K_0, K_1]\text{ for all }q\in(0,1)
    \right\}.
\end{align*}  

\begin{proposition}\label{prop:dist}
For each $s\in\mathcal T$, suppose that $Y_i(s)$ is a measurable random element of 2-Wasserstein space $\mathcal Y$, and $Q_{Y_i(s)} \in \mathcal Q_{[K_0,K_1]}$ almost surely.
Suppose also that $\mathcal T$ is a discrete, finite set and $\P{D_i = d}>0$ for every $d \in \mathcal T$.
    \begin{description}
        \item[(i)] Suppose that Assumption \ref{assumption:D-MTR} holds true and that $K_0,K_1$ are known. Then, the uniformly sharp lower and upper bounds on $\E{Q_{Y_i(t)}(\cdot)}$ are given by 
        \begin{align*}
            \underline Q_{\mathtt{MTR}}(q;t)
            &=
            \E{Q_{Y_i}(q)\mathbf{1}\{D_i\le t\}}
            +
            K_0\P{D_i>t},
            \\
            \overline Q_{\mathtt{MTR}}(q;t)
            &=
            \E{Q_{Y_i}(q)\mathbf{1}\{D_i\ge t\}}
            +
            K_1\P{D_i<t},
        \end{align*}
        respectively.
        
        \item[(ii)] Suppose that Assumption \ref{assumption:D-MTS} holds true and that $K_0,K_1$ are known. Then, the uniformly sharp lower and upper bounds on $\E{Q_{Y_i(t)}(\cdot)}$ are given by 
        \begin{align*}
            \underline Q_{\mathtt{MTS}}(q;t)
            &=
            \E{Q_{Y_i}(q) \mid D_i = t}\P{D_i \geq t} + K_0 \P{D_i < t},
            \\
            \overline Q_{\mathtt{MTS}}(q;t)
            &=
            \E{Q_{Y_i}(q) \mid D_i = t}\P{D_i \leq t} + K_1 \P{D_i > t},
        \end{align*}
        respectively.

        \item[(iii)] Suppose that Assumption \ref{assumption:D-MIV} holds true and that $K_0,K_1$ are known.
        Suppose also that $\mathcal Z$ is a discrete, finite set, and $\P{D_i=d,Z_i=z}>0$ for every $(d,z) \in \mathcal T \times \mathcal Z$.
        For $v=(v_z)_{z\in\mathcal Z}$ with $v_z \in \mathcal Q_{[K_0,K_1]}$, set
        \begin{align*}
            g_z(q;t,v) = 
            \E{Q_i(q)\mathbf{1}\{D_i=t\} \mid Z_i=z} + 
            \P{D_i\neq t\mid Z_i=z} v_z(q).
        \end{align*}
        Define
        \begin{align*}
            \mathcal V_t 
            =
            \left\{
                v=(v_z)_{z\in\mathcal Z}:
                \begin{array}{l}
                    v_z\in\mathcal Q_{[K_0,K_1]}
                    \quad\text{for all }z\in\mathcal Z,\\
                    g_{z_1}(q;t,v)
                    \le
                    g_{z_2}(q;t,v)
                    \quad
                    \text{for all }z_1\le z_2,\ q\in(0,1)
                \end{array}
            \right\}.
        \end{align*}
        Moreover, define the pointwise lower and upper envelopes by
        \begin{align*}
            \ell_z^t(q) = 
            \inf_{v\in\mathcal V_t}v_z(q),
            \quad
            u_z^t(q) =
            \sup_{v\in\mathcal V_t}v_z(q),
        \end{align*}
        for every $z\in\mathcal Z$ and $q\in(0,1)$, and let $\ell^t=(\ell_z^t)_{z\in\mathcal Z}$ and $u^t=(u_z^t)_{z\in\mathcal Z}$.
        Then, the uniformly sharp lower and upper bounds on $\E{Q_{Y_i(t)}(\cdot)}$ are given by
        \begin{align*}
            \underline Q_{\mathtt{MIV}}(q;t)
            =
            \sum_{z\in\mathcal Z} g_z(q;t,\ell^t)\P{Z_i=z},
        \end{align*}
        and
        \begin{align*}
            \overline Q_{\mathtt{MIV}}(q;t)
            =
            \sum_{z\in\mathcal Z} g_z(q;t,u^t)\P{Z_i=z},
        \end{align*}    
        respectively.
    \item[(iv)] Suppose that Assumptions \ref{assumption:D-MTR}--\ref{assumption:D-MTS} hold true. Then, the uniformly sharp lower and upper bounds on $\E{Q_{Y_i(t)}(\cdot)}$ are given by 
        \begin{align*}
            \underline Q_{\mathtt{MTR+MTS}}(q;t)
            &=
            \sum_{d<t}
            \E{Q_i(q)\mid D_i=d}\P{D_i=d} + \E{Q_i(q)\mid D_i=t}\P{D_i\ge t},
            \\
            \overline Q_{\mathtt{MTR+MTS}}(q;t)
            &=
            \sum_{d>t}
            \E{Q_i(q)\mid D_i=d}\P{D_i=d} + \E{Q_i(q)\mid D_i=t}\P{D_i\le t},
        \end{align*}
        respectively. 
    \end{description}
\end{proposition}

Proposition~\ref{prop:dist} shows that, for distributional outcomes, the coordinatewise bounds are uniformly sharp under the MTR, MTS, and joint MTR--MTS assumptions. Consequently, under each of these assumptions, the band between the lower and upper envelopes is the tightest band containing all mean quantile functions consistent with the observed data and the maintained assumptions. Furthermore, the uniform sharpness is stronger than pointwise sharpness: each bounding function is itself attainable under a counterfactual model. This result provides a formal justification for reporting the coordinatewise bounds in empirical applications. Moreover, at each quantile index, these bounds coincide with the corresponding \citeauthor{Manski_Pepper:2000} identified set, making them straightforward to compute. 

This equivalence, however, does not extend to the MIV case. To see the source of non-sharpness, recall that the identified sets in Propositions~\ref{prop:gen-simple} and \ref{prop:dist} are both motivated by the decomposition
\begin{align*}
    \E{Q_{Y_i(t)}(q) \mid Z_i=z} &= 
    \E{Q_i(q)\mathbf{1}\{D_i=t\} \mid Z_i=z} + 
    \P{D_i\neq t\mid Z_i=z} v_z(q)\\
    &= 
    \E{Q_i(q)\mid Z_i=z, D_i = t}\P{D_i=t \mid Z_i=z} + 
    \P{D_i\neq t\mid Z_i=z} v_z(q).
\end{align*}
where $v_z(q)=\E{Q_{Y_i(t)}(q) \mid D_i \neq t, Z_i = z}$.
The pointwise construction of the lower bound first replaces the unknown function $v_z(q)$ by the lower support point $K_0$, and then forms a monotone lower envelope by taking the supremum over instrument values below the target instrument value in order to enforce the MIV restriction for $\E{Q_{Y_i(t)}(q) \mid Z_i=z}$.
The key observation here is that the observed component in the first term need not be monotone in $z$. The missing component $v_z$ must therefore offset any nonmonotonicity in the observed component so that the sum satisfies MIV.
Moreover, the required adjustment may vary with the quantile index $q$. For example, the observed component may be decreasing in $z$ at lower quantile indices but increasing in $z$ at higher quantile indices. In such a case, $v_z(q)$ must be sufficiently large at the lower quantile indices to restore the MIV restriction.
Because $v_z(\cdot)$ must itself be a quantile function, it must be nondecreasing in $q$. Consequently, a large value of $v_z(q)$ required at a lower quantile index also restricts $v_z(q^\prime)$ from below for every $q^\prime>q$. The pointwise construction ignores this cross-quantile restriction and may therefore produce a lower-bound function that cannot be attained simultaneously over all $q$. In this sense, the pointwise lower bound may be conservative relative to the uniformly sharp lower bound. 
The construction of the uniformly sharp bound in Proposition \ref{prop:dist} explicitly takes this restriction into account, where $\mathcal V_t$ imposes both the MIV restrictions on $g_z(q;t,v)$ and the requirement that each $v_z$ be a valid quantile function.

This issue does not arise under MTS, which has a simpler structure.
In particular,
\begin{align*}
    \E{Q_{Y_i(t)}(q) \mid D_i=d} &= 
    \E{Q_i(q)\mathbf{1}\{D_i=t\} \mid D_i=d} + 
    \P{D_i\neq t\mid D_i=d} v_d(q)\\
    &= 
    \begin{cases}
        \E{Q_i(q)\mid D_i = t} & \text{ if }  d=t,\\
        v_d(q) & \text{ if }  d\neq t.
    \end{cases}
\end{align*}
That is, under MTS, the observed and missing components are separated.
Therefore, unlike in the general MIV case, no quantile-dependent adjustment of the missing component is required under MTS. This separation between the observed and missing components allows the pointwise construction to attain the bound uniformly over $q$, yielding the uniformly sharp bound.

\section{Interval outcomes}\label{sec:interval}

In this section, we apply the general framework developed in Section~\ref{sec:gen} to interval-valued outcomes. %Such outcomes arise naturally when variables are observed only within intervals, for example because survey responses are reported in brackets or measurements are subject to interval censoring. 
Interval-valued observations are common in both survey and administrative data. Variables such as household income, earnings, hours worked, and time-use measures are often recorded as intervals rather than exact values (e.g., \citealp{Manski_Tamer:2002}). Such observations appear in numerous data sources, including population censuses and national surveys from Australia, Colombia, Japan, and the United Kingdom \citep{Walter21}.
We show that the coordinatewise bounds developed in Section~\ref{sec:gen} yield the sharp outer interval for interval-valued outcomes.

Let $\mathcal Y$ denote the space of compact intervals on the real line. Following \cite{kuri:25}, we equip $\mathcal Y$ with the metric induced by the support function,
\[
d_{\mathcal Y}(Y_1,Y_2)
=
\left(
\sum_{p\in\mathbb S^0}
\{s(p,Y_1)-s(p,Y_2)\}^2
\right)^{1/2},
\]
where $\mathbb S^0=\{-1,1\}$ and $s(p,Y)=\sup_{y\in Y}\langle p,y\rangle$. Let $\Psi(Y)=s(\cdot,Y)$. Then
\[
d_{\mathcal Y}(Y_1,Y_2)
=
\|\Psi(Y_1)-\Psi(Y_2)\|_{L^2(\mathbb S^0)},
\]
so that $(\mathcal Y,d_{\mathcal Y})$ is isometrically embedded into
$L^2(\mathbb S^0)$. Moreover, $\Psi(\mathcal Y)$ is a closed convex subset of
$L^2(\mathbb S^0)$, as shown in \citet[Lemma 2.1]{kuri:25}. Consequently, the general framework developed in Section~\ref{sec:gen} applies directly with the support function as the embedding map.

The choice of the metric $d_{\mathcal Y}$ is motivated by the fact that, as shown in \citet[Proposition 2.2]{kuri:25}, the (conditional) Fr\'echet mean under $d_{\mathcal Y}$ coincides with the (conditional) Aumann mean, which is a widely used notion of expectation for random sets in the partial identification literature \citep{BeMo08,MoMo18}. Consequently, the target parameter $\mathbb E_\oplus[Y_i(t)]$ coincides with the Aumann mean, which is the standard notion of expectation for random sets.

Now, the identification problem can be formulated through the embedded support functions.
For interval-valued outcomes, the support function can be simplified as
\begin{align*}
    s(p,Y)
    =
    \sup_{y\in Y}\langle p,y\rangle
    =
    \begin{cases}
        -L & \text{if } p=-1,\\
        U & \text{if } p=1,
    \end{cases}
\end{align*}
where $Y=[L,U]$. Consequently, the coordinatewise monotonicity restrictions developed in Section~\ref{sec:gen} reduce to monotonicity restrictions on the interval endpoints. 
This observation motivates the following interval-valued analogue of the MTR assumption.

\begin{assumption}[MTR for interval outcomes]\label{assumption:I-MTR}
    Let $\mathcal{T}$ be an ordered set. For all $t_1,t_2\in\mathcal T$,
    \begin{align*}
        t_2\ge t_1\Longrightarrow L_i(t_2)\ge L_i(t_1),\quad U_i(t_2)\ge U_i(t_1),
    \end{align*}
    almost surely, where $Y_i(t)=[L_i(t), U_i(t)]$.
\end{assumption}
Assumption \ref{assumption:I-MTR} requires endpointwise monotone treatment response.
Using the formulation used in Section \ref{sec:gen}, this condition can be rewritten as
\begin{align*}
    t_2 \ge t_1 \Longrightarrow s(1,Y_i(t_2)) \ge s(1,Y_i(t_1))\,\text{ and }\,
    s(-1,Y_i(t_2)) \le s(-1,Y_i(t_1)),
\end{align*}
that is, $\mathcal K_+ = \{1\}$, $\mathcal K_- = \{-1\}$, and $\mathcal K_0 = \varnothing$.
In the educational attainment example, this assumption states that if the worker's years of schooling increases exogenously, the interval-measured wage shifts weakly upward.

We next extend the MTS assumption to interval-valued outcomes by imposing monotonicity on the conditional expectations of the interval endpoints. This is the interval-valued analogue of the scalar MTS assumption.

\begin{assumption}[MTS for interval outcomes]\label{assumption:I-MTS}
    Let $\mathcal{T}$ be an ordered set. 
    For each $t\in\mathcal T$ and all $t_1,t_2\in\mathcal T$,
    \begin{align*}
        t_2 \ge t_1 \Longrightarrow \E{L_i(t) | D_i=t_2} \ge \E{L_i(t) | D_i=t_1},\quad
        \E{U_i(t) | D_i=t_2} \ge \E{U_i(t) | D_i=t_1}.
    \end{align*}
\end{assumption}

We similarly extend the MIV assumption by requiring monotonicity of the conditional expectations of the interval endpoints with respect to the instrumental variable.

\begin{assumption}[MIV for interval outcomes] \label{assumption:I-MIV}
    Let $\mathcal{Z}$ be an ordered set. 
    For each $t\in\mathcal{T}$ and all $(z_1, z_2) \in \mathcal{Z}\times\mathcal{Z}$ such that $z_2\ge z_1$, 
    \begin{align*}
        \E{L_i(t) \mid Z_i = z_2} \ge \E{L_i(t) \mid Z_i = z_1},\quad 
        \E{U_i(t) \mid Z_i = z_2} \ge \E{U_i(t) \mid Z_i = z_1}.
    \end{align*}
\end{assumption}

For interval-valued outcomes, the parameter of interest is itself an interval. Accordingly, sharpness can be summarized by the smallest interval in the sharp identified set that contains every other identified interval. The following definition formalizes this notion.

\begin{definition}
    Let $\mathcal{S}$ be the sharp identified set of an interval-valued target parameter. Then the interval $I$ is called the \textit{sharp outer interval} if (i) $[a,b]\subseteq I$ for every interval $[a,b]\in\mathcal{S}$ and (ii) $I\in\mathcal{S}$.
\end{definition}

When it exists, the sharp outer interval is uniquely determined by the smallest attainable lower endpoint and the largest attainable upper endpoint. The following proposition shows that this interval exists under the interval-valued monotonicity assumptions and admits a simple closed-form expression. 

\begin{proposition}\label{prop:int}
    For each $s\in\mathcal T$, suppose that $K_0 \leq L_i(s) \le U_i(s)\leq K_1$ almost surely, $\mathcal T$ and $\mathcal Z$ are discrete, finite, and $\P{D_i=d}>0$ for every $d\in\mathcal T$. 

    \begin{description}
         \item[(i)] Suppose that $K_0$ and $K_1$ are known. Under Assumption \ref{assumption:I-MTR}, the sharp outer interval of $\mathbb E_\oplus[Y_i(t)]$ is
         \begin{align*}
            I_{\mathtt{MTR}}(t) = \bigg[
            \E{L_i\mathbf{1}\{D_i\le t\}} + K_0\P{D_i>t},
            \,
            \E{U_i\mathbf{1}\{D_i\ge t\}} + K_1\P{D_i<t}
            \bigg].
        \end{align*}

        \item[(ii)] Suppose that $K_0$ and $K_1$ are known. Under Assumption \ref{assumption:I-MTS}, the sharp outer interval of $\mathbb E_\oplus[Y_i(t)]$ is
        \begin{align*}
            I_{\mathtt{MTS}}(t)
            =
            \bigg[
            \E{L_i\mid D_i=t}\P{D_i\ge t} + K_0\P{D_i<t},
            \,
            \E{U_i\mid D_i=t}\P{D_i\le t} + K_1\P{D_i>t}
            \bigg].
        \end{align*}

        \item[(iii)] Suppose that $K_0$ and $K_1$ are known and that $\P{D_i=d,Z_i=z}>0$ for every $(d,z) \in \mathcal T \times \mathcal Z$. Under Assumption \ref{assumption:I-MIV}, the sharp outer interval of $\mathbb E_\oplus[Y_i(t)]$ is
        \begin{align*}
            I_{\mathtt{MIV}}(t)
            =
            \bigg[
            &
            \sum_{z\in\mathcal Z}\P{Z_i=z}
            \sup_{z^\prime\le z}
            \left\{
            \mathbb{E}[L_i\mathbf 1\{D_i=t\}\mid Z_i=z^\prime]
            +
            K_0\P{D_i\ne t\mid Z_i=z^\prime}
            \right\},\\
            &\quad
            \sum_{z\in\mathcal Z}\P{Z_i=z}
            \inf_{z^\prime\ge z}
            \left\{\mathbb{E}[U_i\mathbf 1\{D_i=t\}\mid Z_i=z^\prime]+K_1\P{D_i\ne t\mid Z_i=z^\prime}\right\}
            \bigg].
        \end{align*}

        \item[(iv)] Under Assumptions \ref{assumption:I-MTR} and \ref{assumption:I-MTS}, the sharp outer interval of $\mathbb E_\oplus[Y_i(t)]$ is
        \begin{align*}
            I_{\mathtt{MTR+MTS}}(t)
            =
            \bigg[
            &
            \sum_{d<t}
            \E{L_i\mid D_i=d}\P{D_i=d}
            +
            \E{L_i\mid D_i=t}\P{D_i\ge t},
            \\
            &\quad
            \sum_{d>t}
            \E{U_i\mid D_i=d}\P{D_i=d}
            +
            \E{U_i\mid D_i=t}\P{D_i\le t}
            \bigg].
        \end{align*}
    \end{description}
\end{proposition}

%\red{[From Daisuke: We need to add a discussion on the interpretations of $\sum_{z \in \mathcal{Z}}, \sum_{d>t}$ when the ordered sets $\mathcal{Z},\mathcal{T}$ has infinitely many elements.]}
Proposition~\ref{prop:int} shows that the sharp outer interval admits a simple closed-form expression under each monotonicity assumption. As in the distributional outcome case, its endpoints coincide with those of the conventional \citeauthor{Manski_Pepper:2000} identified set. Thus,
even when outcomes are interval-valued rather than scalar, the analysis proceeds in essentially the same way as in the scalar case. This result broadens the applicability of monotonicity-based identification methods to a wide range of economic settings involving interval-valued data. The resulting bounds may, however, be less informative than those obtained when outcomes are observed exactly. The extent of this loss of informativeness is examined in the following section through a numerical application.

\section{Empirical illustrations} \label{sec:numerical}

This section illustrates the proposed identification framework through two empirical applications involving distribution-valued and interval-valued outcomes.

\subsection{Distributional outcomes: NHANES periodontal health}

This subsection examines the effect of heavy smoking on the distribution of clinical attachment loss (CAL), a measure of the loss of periodontal support around a tooth.

Oral health is important not only for maintaining a healthy diet and immune function, but also for labor-market outcomes and self-esteem (e.g., \citealp{Gallego_etal:2024, Gleid10}). Smoking has long been recognized as a major risk factor for periodontal disease \citep{Leite18}. We therefore investigate the effect of heavy smoking on the distribution of CAL.

Our analysis uses data from the 2013--2014 National Health and Nutrition Examination Survey (NHANES). The survey contains basic demographic information, responses to health-related questions, including the frequency and intensity of cigarette consumption, and the results of oral health examinations conducted by dentists licensed in at least one U.S. state.
In particular, NHANES records CAL at six sites per tooth for up to 28 teeth, wherever assessment is feasible. Using these up to 168 site-level measurements, we construct an individual-level distribution of CAL, denoted by $Y_i$, which is our main outcome variable.

As the treatment variable, we classify respondents into three groups: non-smokers ($D_i=0$); non-heavy current smokers ($D_i=1$), defined as those who reported smoking fewer than 20 cigarettes per day during the previous 30 days; and heavy current smokers ($D_i=2$), defined as those who reported smoking 20 or more cigarettes per day during the previous 30 days. 
The threshold of 20 cigarettes per day follows the definition in \citet[p.~30]{statsmoke}.
We exclude former smokers and respondents whose smoking histories do not fall into these three categories. We are interested in the counterfactual average CAL distribution for nonsmokers $\mathbb{E}_{\oplus}[Y_i(0)]$ and heavy-smokers $\mathbb{E}_{\oplus}[Y_i(2)]$.
Equivalently, we estimate the mean quantile function $\E{Q_{Y_i(t)}(q)}$ for $q\in(0,1)$ and $t\in\{0,2\}$.

We impose the MTR and MTS assumptions; that is, Assumptions \ref{assumption:D-MTR} and \ref{assumption:D-MTS}. MTR posits that the distribution of CAL under a higher level of smoking is weakly worse than that under a lower level of smoking. Since a larger value of CAL indicates greater periodontal destruction, this assumption implies that the counterfactual CAL distribution shifts weakly toward more severe periodontal outcomes as smoking intensity increases. 
MTS assumes that, on average, individuals in higher observed smoking categories would have weakly worse counterfactual CAL distributions than individuals in lower observed smoking categories, under any common treatment level. This assumption may be plausible if heavier smokers are, on average, less likely to engage in behaviors that promote oral health.

We first compare the identifying power of MTR alone with that of the combined MTR--MTS assumptions.
Figure~\ref{fig:compare} presents the identified region for $\E{Q_{Y_i(0)}(q)}$ and $\E{Q_{Y_i(2)}(q)}$. 
The uniformly sharp bounds under the joint MTR--MTS assumptions are considerably narrower than those under MTR alone, indicating that MTS provides substantial additional identifying power. This improvement is particularly pronounced in Figure~\ref{fig:compare2}.
\begin{figure}[t]
    \centering
    \begin{subfigure}[b]{0.49\textwidth}
        \centering
        \includegraphics[width=\linewidth]{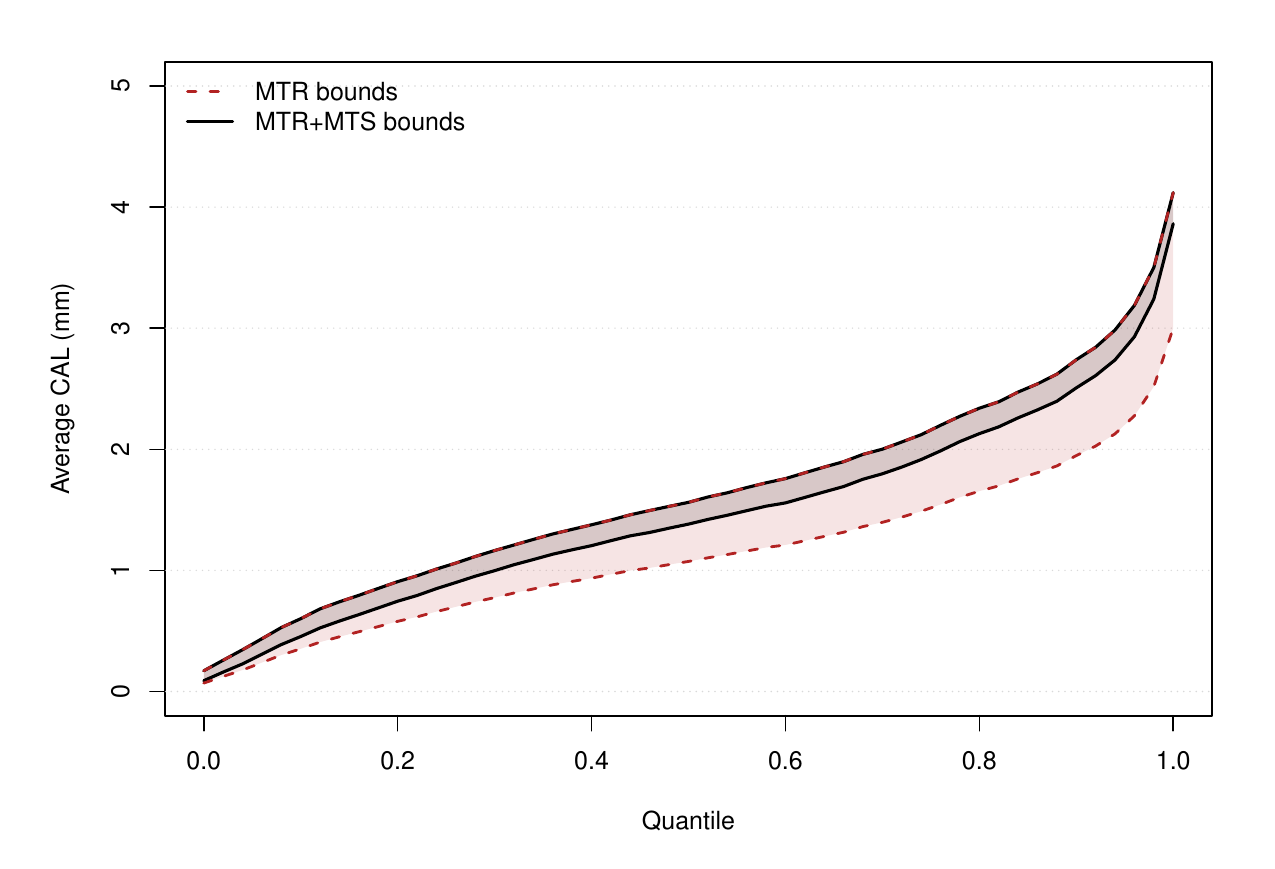}
        \caption{MTR vs. MTR+MTS ($D_i=0$)}
        \label{fig:compare1}
    \end{subfigure}
    \hfill
    \begin{subfigure}[b]{0.49\textwidth}
        \centering
        \includegraphics[width=\linewidth]{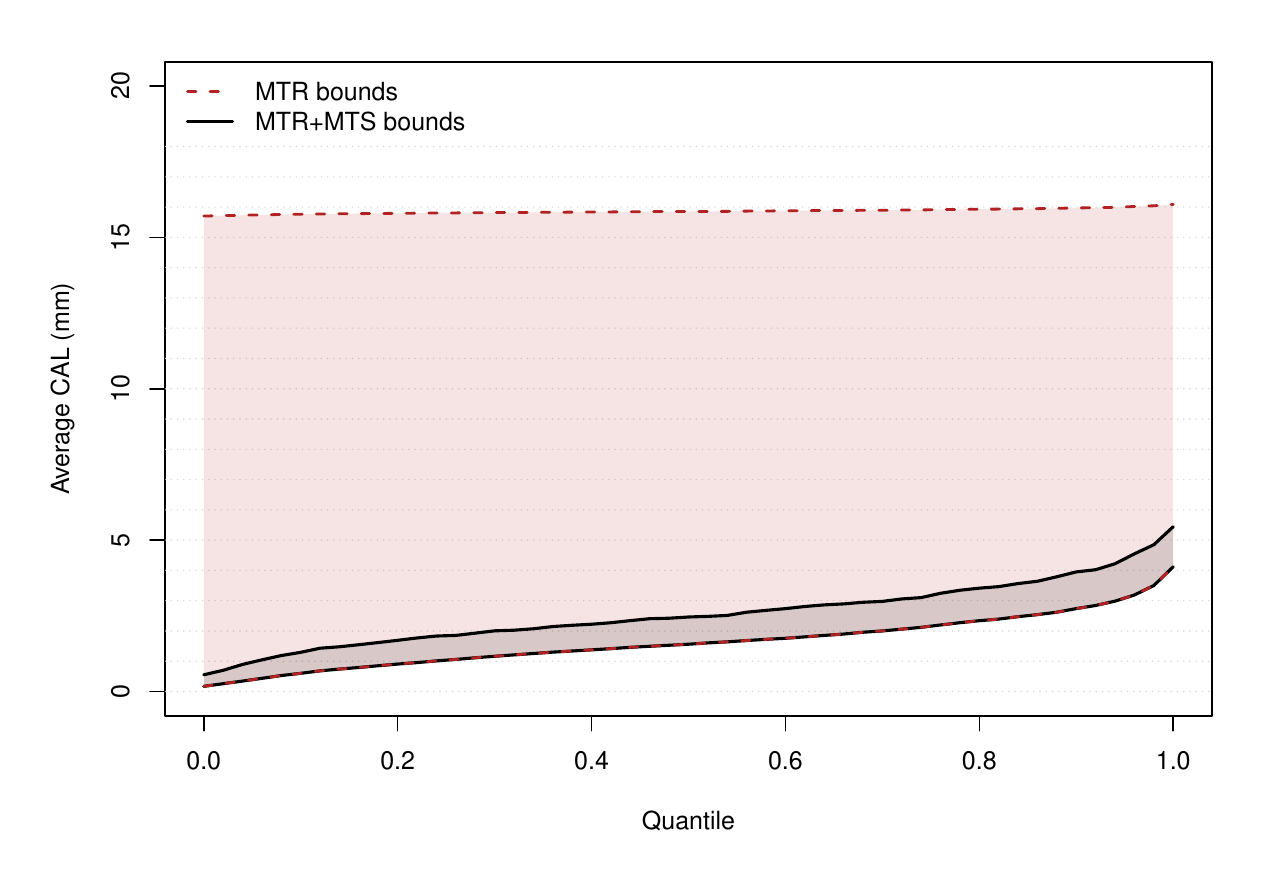}
        \caption{MTR vs. MTR+MTS ($D_i=2$)}
        \label{fig:compare2}
    \end{subfigure}
    \caption{Identifying power of MTR and MTR+MTS assumptions}
    \label{fig:compare}

    \begin{flushleft}
    \footnotesize
    \renewcommand{\baselineskip}{11pt}
    \textbf{Note:} To compute the MTR bounds, we set the lower and upper support bounds to $K_0=0$ and $K_1=17$, respectively. The lower bound follows from the fact that CAL is nonnegative, whereas the upper bound is set equal to the maximum value specified in the NHANES variable documentation.
    \end{flushleft}
\end{figure}

Figure \ref{fig:te1} compares the MTR--MTS bounds for $\E{Q_{Y_i(0)}(q)}$ and $\E{Q_{Y_i(2)}(q)}$. Under these assumptions, the counterfactual CAL distribution under heavy smoking is weakly shifted upward relative to that under non-smoking. More concretely, the mean quantile function under heavy smoking lies uniformly above that under non-smoking. Figure~\ref{fig:te2} reports the implied average quantile-wise treatment effects. The upper bounds indicate that heavy smoking could increase CAL by as much as $0.5$ mm at lower quantiles and $1.5$ mm at upper quantiles. 
\begin{figure}[t]
    \centering
    \begin{subfigure}[b]{0.49\textwidth}
        \centering
        \includegraphics[width=\linewidth]{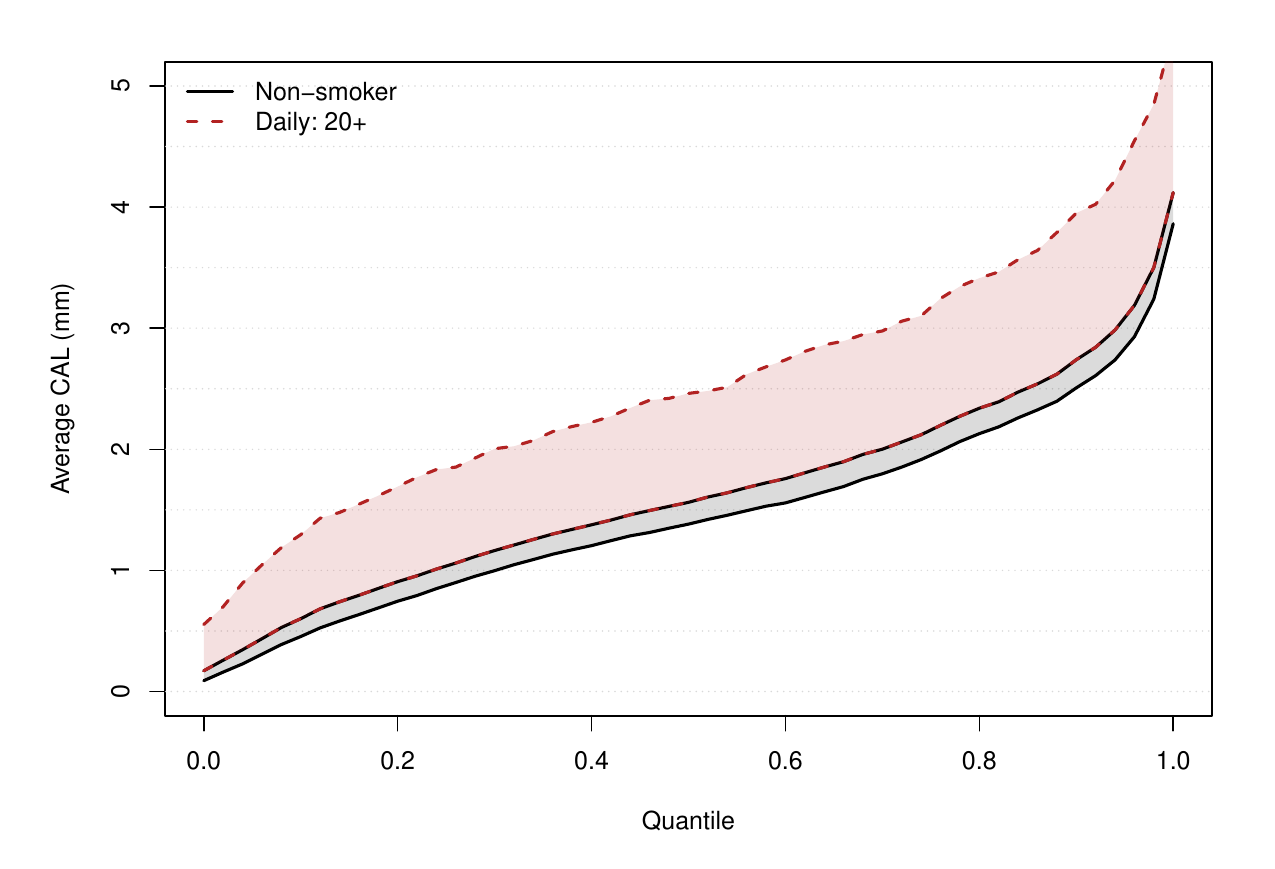}
        \caption{MTR+MTS bounds for $D_i\in\{0,2\}$}
        \label{fig:te1}
    \end{subfigure}
    \hfill
    \begin{subfigure}[b]{0.49\textwidth}
        \centering
        \includegraphics[width=\linewidth]{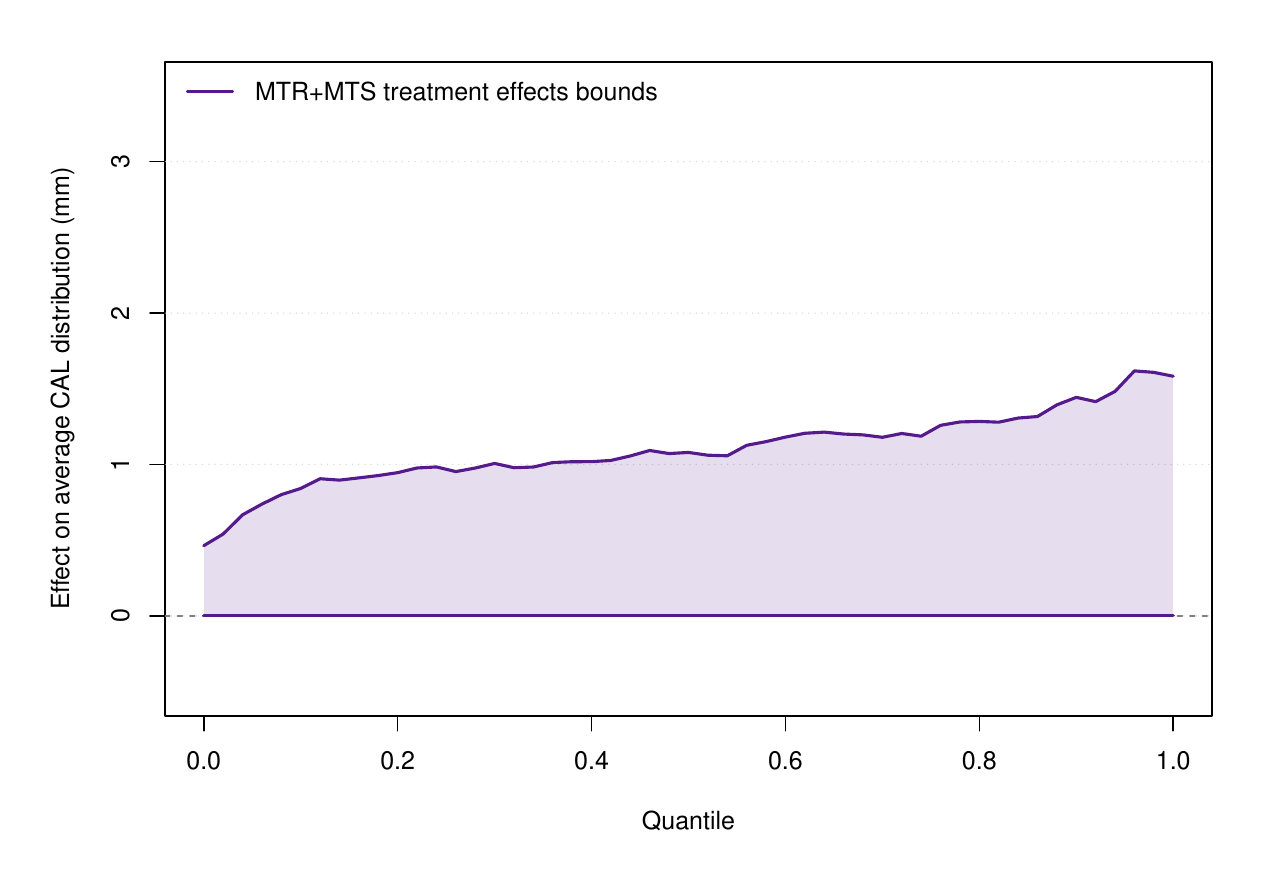}
        \caption{Averaged quantile treatment effects bounds}
        \label{fig:te2}
    \end{subfigure}
    \caption{Counterfactual mean and treatment effects}
    \label{fig:te}

    %\begin{flushleft}
    %\footnotesize
    %\renewcommand{\baselineskip}{11pt}
    %\textbf{Note:} 
    %\end{flushleft}
\end{figure}

Finally, we examine a testable implication of the MTR and MTS assumptions. Under these assumptions, for $t_1\leq t_2$,
\begin{align*}
    \mathbb{E}[Q_{Y_i}(q) \mid D_i = t_1]
    &= \mathbb{E}[Q_{Y_i(t_1)}(q) \mid D_i = t_1] \leq \mathbb{E}[Q_{Y_i(t_2)}(q) \mid D_i = t_1] \leq \mathbb{E}[Q_{Y_i(t_2)}(q) \mid D_i = t_2] \\
    &= \mathbb{E}[Q_{Y_i}(q) \mid D_i = t_2],
\end{align*}
where the first inequality follows from MTR and the second from MTS. Hence, the observed conditional mean quantile functions must be weakly ordered across smoking categories. Figure~\ref{fig:testable} plots $\mathbb{E}[Q_{Y_i}(q) \mid D_i = d]$ for $d\in\{0,1,2\}$. The quantile functions appear to satisfy this ordering, indicating that the assumptions are visually consistent with the observed data.

\begin{figure}
    \centering
    \includegraphics[width=0.65\linewidth]{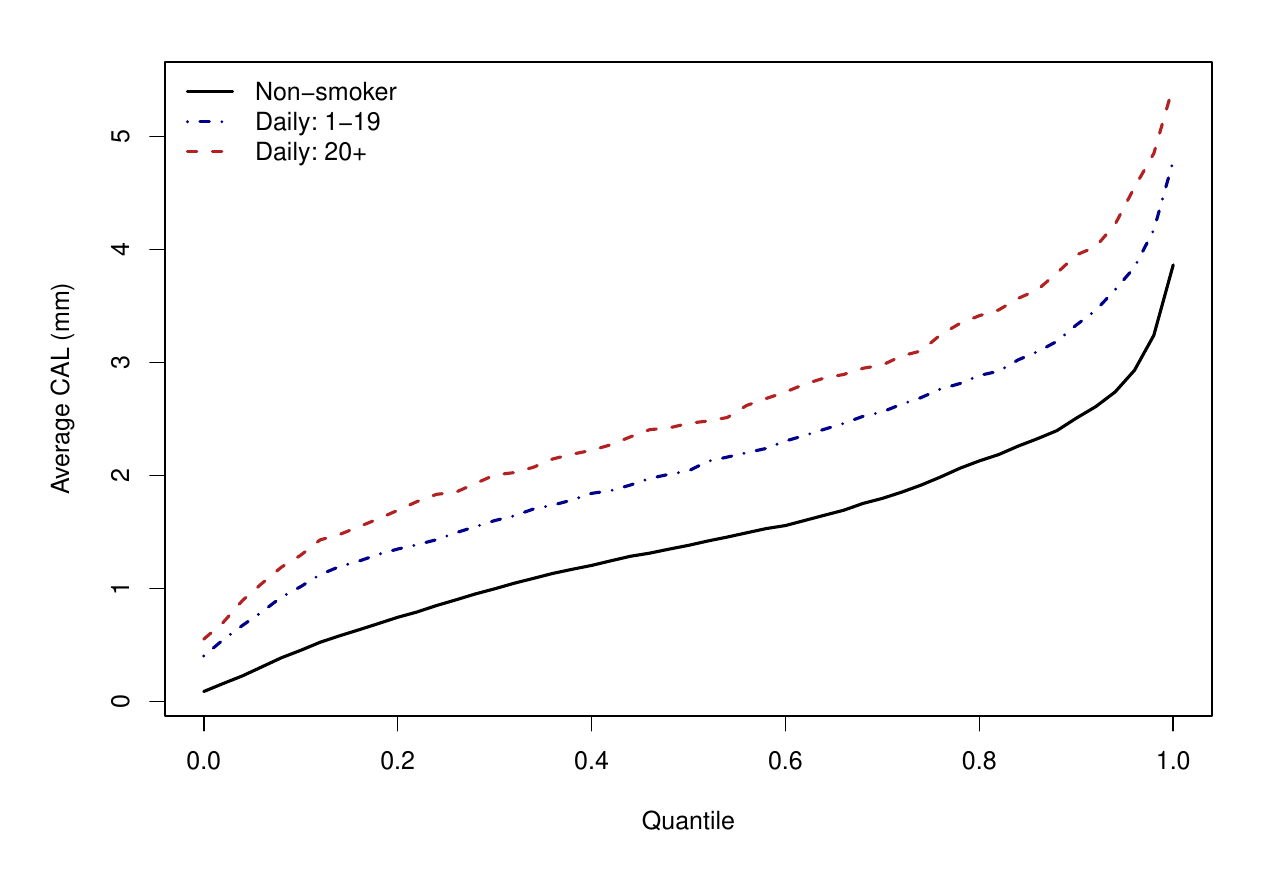}
    \caption{Average observed quantile functions}
    \label{fig:testable}
    %\begin{flushleft}
    %\footnotesize
    %\renewcommand{\baselineskip}{11pt}
    %\textbf{Note:} 
    %\end{flushleft}
\end{figure}

\subsection{Interval-valued outcomes: Job Corps earnings}

We illustrate the empirical implications of our identification results for interval-valued outcomes using the Job Corps application studied by \cite{Kim:2018}. The study examines counterfactual mean weekly earnings in the sixteenth quarter after treatment assignment. We use the same analysis sample and treatment definition as in \cite{Kim:2018}, with treatment measured by the number of weeks of training received.
The dataset is obtained from the replication package of \cite{Kim:2018}; its original source is the National Job Corps Study analyzed by \cite{Schochet08}.

Our objective here is not to argue that the original Job Corps earnings measure is itself interval-valued. Rather, we conduct a counterfactual coarsening exercise. We replace the original point-valued outcome with an interval-valued observation by applying coarsening rules used in existing surveys and censuses. We then apply our identification procedure to the resulting interval-valued outcome. The bounds computed from the original point-valued outcome provide an empirical benchmark. Comparing our bounds with this benchmark therefore quantifies the loss of identifying information induced by alternative reporting and data-release designs.

We consider two empirically motivated coarsening rules.
The first rule is employed in recent public-use files from the U.S. Current Population Survey (CPS). In CPS, zero weekly earnings remain observed exactly; positive earnings of $7$ or less are released as $4$; earnings above $7$ and below $1,200$ are rounded to the nearest multiple of $2$; earnings from $1,200$ to below $2,000$ are rounded to the nearest multiple of $10$, and earnings of $2,000$ or more are rounded to the nearest multiple of $20$. The corresponding interval-valued observations are given by the set of values that could have generated each released amount. For example, the released values $0$, $4$, $218$, $1{,}450$, and $2{,}020$ correspond to the intervals $\{0\}$, $[0,7]$, $[217,219]$, $[1{,}445,1{,}455]$, and $[2{,}010,2{,}030]$, respectively. At the two threshold values, $1{,}200$ and $2{,}000$, we use the full interval implied by the adjoining rounding rules to avoid ambiguity in the induced intervals.

Our second rule is based on the weekly personal-income bands in the 2021 Australian Census, which exhibit coarser reporting. 
The original weekly-income cutpoints are
\begin{align*}
0,\ 150,\ 300,\ 400,\ 500,\ 650,\ 800,\ 1{,}000,\ 1{,}250,\ 1{,}500,\ 1{,}750,\ 2{,}000,\ 3{,}000,\ 3{,}500.
\end{align*}
To adapt the numerical scale of these Australian-dollar cutoffs to the U.S. Job Corps earnings distribution, we multiply each positive cutoff by $0.7$, which approximately reflects the Australian-dollar-to-U.S.-dollar exchange rate.  We retain zero weekly earnings as an exact value and assign each positive earnings observation to one of the income bands.

The CPS-style rule and the Australian-Census-style rule provide two distinct benchmarks. The former reflects relatively mild coarsening induced by a public-use data-release procedure. The latter represents a substantially coarser survey-style reporting design, which is also commonly used to improve survey response rates.

For each rule, we recompute the identified set under our proposed MTR and MTS restrictions. Figure~\ref{fig:jobcorps} compares the resulting bounds with the benchmark bounds of \cite{Kim:2018} obtained from the original point-valued earnings measure.
\begin{figure}[t]
    \centering
    \begin{subfigure}[b]{0.49\textwidth}
        \centering
        \includegraphics[width=\linewidth]{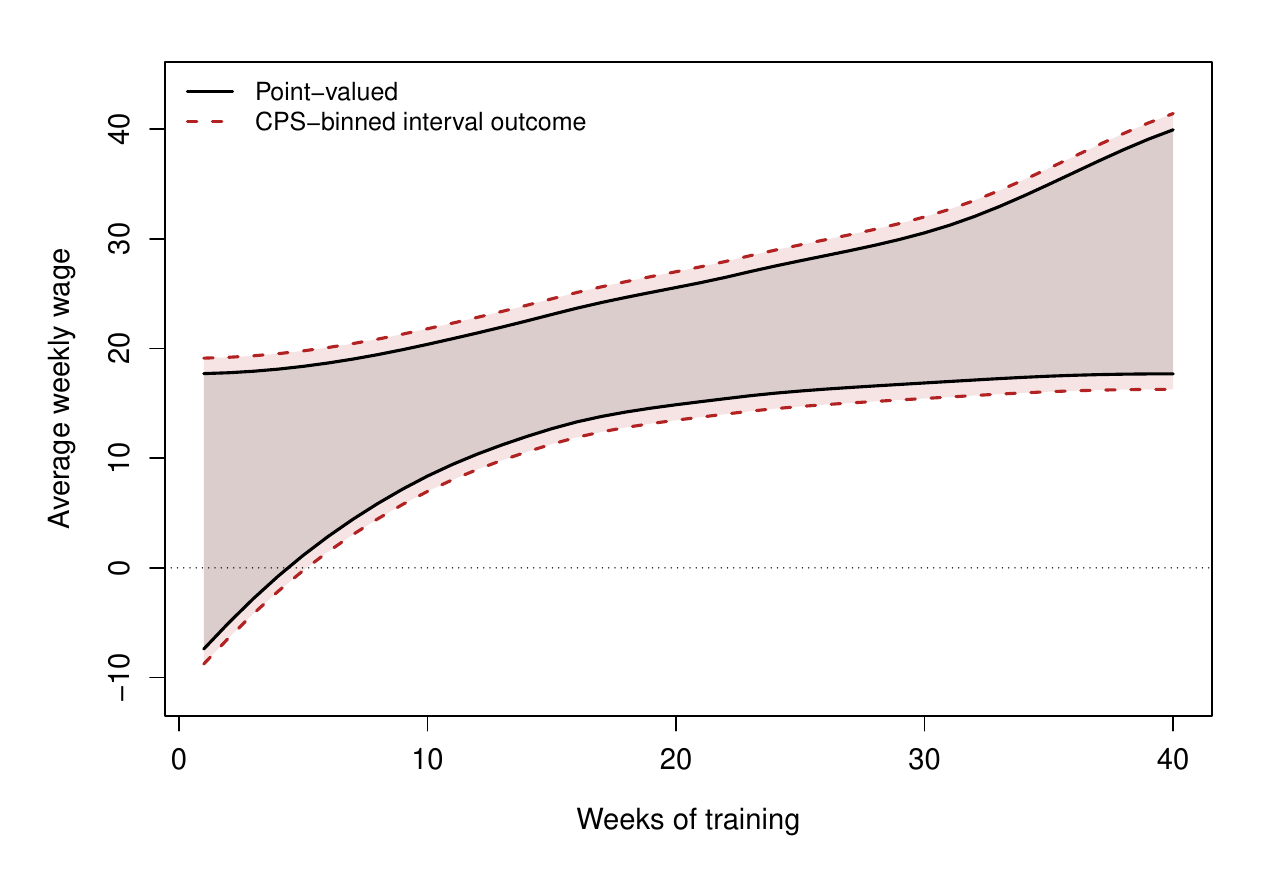}
        \caption{CPS}
        \label{fig:jobcorps1}
    \end{subfigure}
    \hfill
    \begin{subfigure}[b]{0.49\textwidth}
        \centering
        \includegraphics[width=\linewidth]{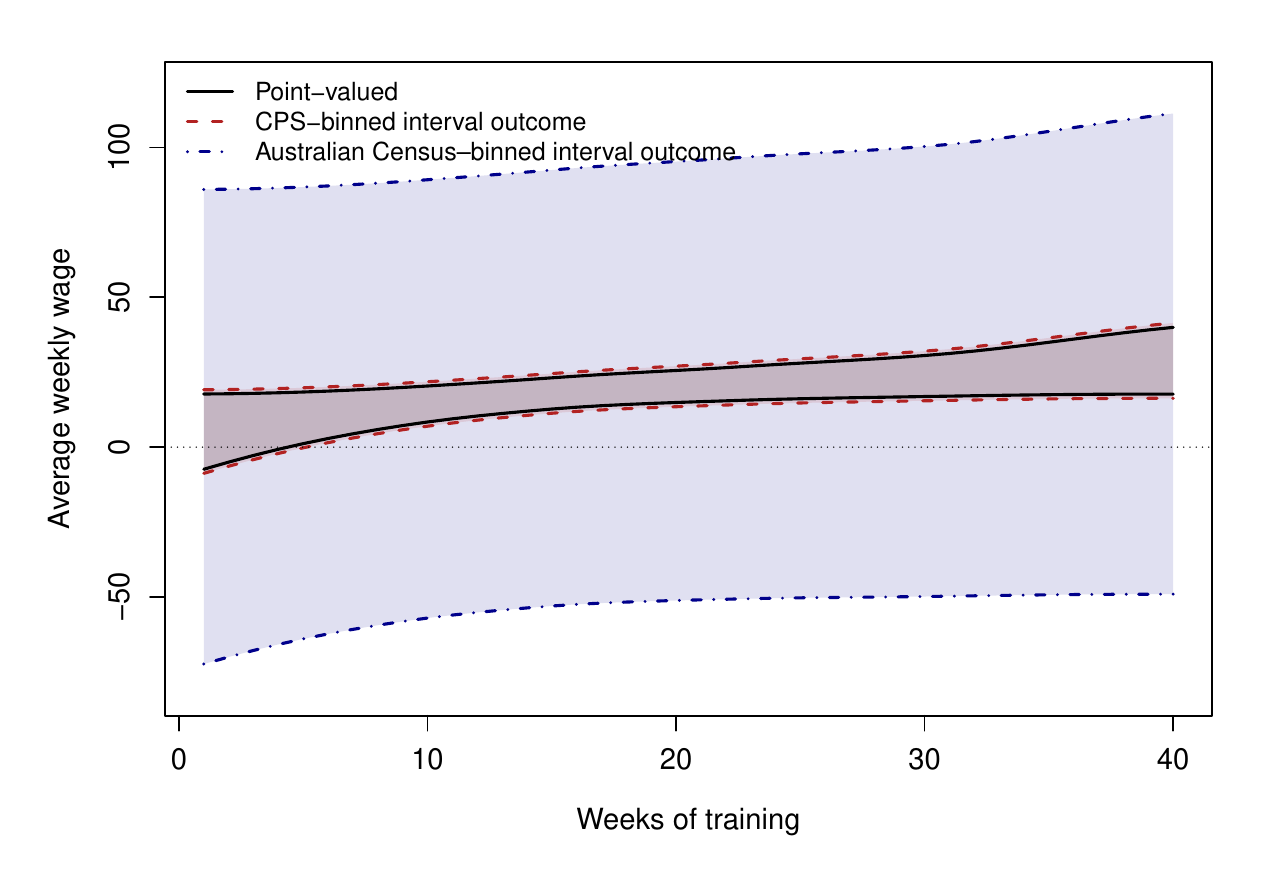}
        \caption{Australian Census}
        \label{fig:jobcorps2}
    \end{subfigure}
    \caption{MTR+MTS bounds for interval outcomes}
    \label{fig:jobcorps}

    \begin{flushleft}
    \footnotesize
    \renewcommand{\baselineskip}{11pt}
    \textbf{Note:} $\mathbb{E}_{\oplus}[Y_i(t)]$ for each treatment duration $t$ is reported.
    \end{flushleft}
\end{figure}

The exercise yields several important takeaways. First, even when the underlying point-valued outcome is unavailable---for example, because confidentiality protection requires the release of coarsened outcomes, as in the CPS---the bounds obtained under our proposed MTR and MTS assumptions can remain highly informative. In our application, the bounds based on the CPS-style coarsening remain close to the benchmark bounds computed from the original point-valued outcome under the corresponding MTR+MTS restrictions.

At the same time, the degree of coarsening is consequential. When outcomes are recorded using intervals as coarse as those in the Australian Census design, the resulting bounds become substantially wider and provide little useful information about the parameter of interest. This contrast highlights an important implication for survey and experimental design. Although confidentiality protection and reduced respondent burden are important objectives, preserving as much outcome granularity as possible can be critical for obtaining informative bounds and, consequently, sharper policy conclusions.

\section{Conclusion} \label{sec:conclusion}

This paper develops a unified framework for partial identification with random-object outcomes under monotone treatment response (MTR), monotone treatment selection (MTS), and monotone instrumental variable (MIV) assumptions. Building on an isometric embedding of random objects into an $L^2$ space, we formulate monotonicity restrictions coordinatewise on the embedded functions and derive identified sets for the corresponding Fr\'echet means. The framework applies to a broad class of random objects through a common representation and accommodates the three classical monotonicity assumptions within a unified framework.

We further establish sharp identification results for three economically important classes of random objects. For compositional outcomes under the Aitchison metric, the coordinatewise bounds are sharp. For distributional outcomes under the Wasserstein metric, the coordinatewise construction leads to uniformly sharp lower and upper bounds under MTR, MTS, and joint MTR--MTS restrictions. For interval-valued outcomes represented by support functions, the corresponding bounds coincide with the sharp outer interval.
The empirical application to periodontal health distributions and the numerical illustration based on Job Corps earnings data demonstrate that the proposed framework can yield informative identified sets in empirically relevant settings.

\newpage

\appendix

\section{General construction of sharp identified sets} \label{sec:gen-set}

This section characterizes the sharp identified set for general random objects. The results below may be useful for analyzing objects that are not explicitly considered in this article and are also used in the proofs of our propositions.

For $x,x^\prime \in \Psi(\mathcal Y)$, we write $x \preceq x^\prime$ when $x(k)\leq x^\prime(k)$ for a.e. $k\in\mathcal K_+$, $x(k)\geq x^\prime(k)$ for a.e. $k\in\mathcal K_-$, and no restriction is imposed on coordinates $k\in\mathcal{K}_0$.
Write $X_i^{D_i} = X_i$.
\begin{lemma}\label{lem:obj}
    Suppose that Assumption \ref{assumption:metric} holds true.
    Suppose also that $\mathcal Y$ is known.
    Assume that $\mathcal K_+$, $\mathcal K_-$, and $\mathcal K_0$ are measurable and form a partition of $\mathcal K$ up to $\nu$-null sets.
    Let $\mathcal T$ and $\mathcal Z$ be discrete, finite sets.
    Finally, assume $\P{D_i = d}>0$ for every $d \in \mathcal T$.
    \begin{description}
        \item[(i)] Suppose that Assumption \ref{assumption:MTR-obj} holds. Define, for each $d\in\mathcal T$,
        \begin{align*}
            \mathcal F_d^{\mathtt{MTR}}
            =
            \left\{
            (f_d^s)_{s\in\mathcal T} \in \Psi(\mathcal Y)^{\mathcal T}:
            \begin{array}{l}
                \text{there exist } (\Omega_d, \mathcal{A}_d, \mathbb{P}_d) \text{ and random elements } (\xi_i^s)_{s\in\mathcal T}
                \text{ such that }\\
                \qquad \mathbb{P}_d[\xi_i^d \in A] = \P{X_i \in A \mid D_i = d},\\
                \qquad s_1 \le s_2 \Rightarrow \xi_i^{s_1} \preceq \xi_i^{s_2}, \mathbb{P}_d\text{-a.s.},\\
                \qquad\mathbb{E}_{\mathbb{P}_d}[\xi_i^s]=f_d^s \text{ for all }s\in\mathcal T,\\
                \qquad\mathbb{E}_{\mathbb{P}_d}[\|\xi_i^s\|_{L^2(\nu)}^2] < \infty \text{ for all }s\in\mathcal T.
            \end{array}
            \right\}.
        \end{align*}
        Then the sharp identified set of $\E{\Psi(Y_i(t))}$ is given by
        \begin{align*}
            \mathcal S_{\mathtt{MTR}}(t)
            =
            \left\{
            \sum_{d\in\mathcal T}\P{D_i=d}\, f_d^t:
            (f_d^s)_{s\in\mathcal T}
            \in
            \mathcal F_d^{\mathtt{MTR}}
            \text{ for all }d\in\mathcal T
            \right\}.
        \end{align*}

        \item[(ii)] Suppose that Assumption \ref{assumption:MTS-obj} holds. 
        Define, for each $d\in\mathcal T$,
        \begin{align*}
            \mathcal F_d^0 = 
            \left\{
            (f_d^s)_{s\in\mathcal T} \in \Psi(\mathcal Y)^{\mathcal T}:
            \begin{array}{l}
                \text{there exist } (\Omega_d, \mathcal{A}_d, \mathbb{P}_d) \text{ and random elements } (\xi_i^s)_{s\in\mathcal T}
                \text{ such that }\\
                \qquad \mathbb{P}_d[\xi_i^d \in A] = \P{X_i \in A \mid D_i = d},\\
                \qquad\mathbb{E}_{\mathbb{P}_d}[\xi_i^s]=f_d^s \text{ for all }s\in\mathcal T,\\
                \qquad\mathbb{E}_{\mathbb{P}_d}[\|\xi_i^s\|_{L^2(\nu)}^2] < \infty \text{ for all }s\in\mathcal T.
            \end{array}
            \right\}.
        \end{align*}
        Then the sharp identified set of $\E{\Psi(Y_i(t))}$ is given by
        \begin{align*}
            \mathcal S_{\mathtt{MTS}}(t) = 
            \left\{ \sum_{d\in\mathcal T}\P{D_i=d}\, f_d^t: 
            \begin{array}{l} 
                (f_d^s)_{s\in\mathcal T}\in\mathcal F_d^0 \text{ for all }d\in\mathcal T,\\ 
                d_1\le d_2 \Rightarrow f_{d_1}^s\preceq f_{d_2}^s \text{ for all }s\in\mathcal T 
            \end{array} 
            \right\}.
        \end{align*}

        \item[(iii)] Suppose that Assumption \ref{assumption:MIV-obj} holds.
        Suppose also that $\P{D_i=d,Z_i=z}>0$ for every $(d,z) \in \mathcal T \times \mathcal Z$.
        Define, for each $(d,z)\in\mathcal T \times \mathcal Z$,
        \begin{align*}
            \mathcal F_{d,z}^0 = \left\{
            (f_{d,z}^s)_{s\in\mathcal T} \in \Psi(\mathcal Y)^{\mathcal T}:
            \begin{array}{l}
                \text{there exist } (\Omega_{d,z}, \mathcal{A}_{d,z}, \mathbb{P}_{d,z}) \text{ and random elements } (\xi_i^s)_{s\in\mathcal T}
                \text{ such that }\\
                \qquad \mathbb{P}_{d,z}[\xi_i^d \in A] = \P{X_i \in A \mid D_i = d, Z_i = z},\\
                \qquad\mathbb{E}_{\mathbb{P}_{d,z}}[\xi_i^s]=f_{d,z}^s \text{ for all }s\in\mathcal T,\\
                \qquad\mathbb{E}_{\mathbb{P}_{d,z}}[\|\xi_i^s\|_{L^2(\nu)}^2] < \infty \text{ for all }s\in\mathcal T.
            \end{array}
            \right\}.
        \end{align*}
        For any collection $f=\{f_{d,z}^s:s\in\mathcal T,\ d\in\mathcal T,\ z\in\mathcal Z\}$, define, for each $s\in\mathcal T$ and $z\in\mathcal Z$,
        \begin{align*}
            m_z^s(f) = \sum_{d\in\mathcal T} \P{D_i=d\mid Z_i=z}\, f_{d,z}^s
        \end{align*}
        Then the sharp identified set of $\E{\Psi(Y_i(t))}$ is given by
        \begin{align*}
            \mathcal S_{\mathtt{MIV}}(t) = \left\{ \sum_{z\in\mathcal Z}\P{Z_i=z}\, m_z^t(f):
            \begin{array}{l} 
                (f_{d,z}^s)_{s\in\mathcal T}\in\mathcal F_{d,z}^0 \text{ for all }(d,z)\in\mathcal T\times\mathcal Z,\\
                z_1\le z_2 \Rightarrow m_{z_1}^s(f)\preceq m_{z_2}^s(f) \text{ for all }s\in\mathcal T 
            \end{array}
            \right\}.
        \end{align*}

        \item[(iv)] Suppose that Assumptions \ref{assumption:MTR-obj} and \ref{assumption:MTS-obj} hold.
        Then the sharp identified set of $\E{\Psi(Y_i(t))}$ is given by
        \begin{align*}
            \mathcal S_{\mathtt{MTR+MTS}}(t) = 
            \left\{ \sum_{d\in\mathcal T}\P{D_i=d}\, f_d^t: 
            \begin{array}{l} 
                (f_d^s)_{s\in\mathcal T} \in \mathcal F_d^{\mathtt{MTR}} \text{ for all }d\in\mathcal T,\\
                d_1\le d_2 \Rightarrow f_{d_1}^s\preceq f_{d_2}^s \text{ for all }s\in\mathcal T 
            \end{array} 
            \right\}.
        \end{align*}
    \end{description}
\end{lemma}

\section{Proofs} \label{sec:proof}

\begin{proof}[Proof of Lemma \ref{lem:obj}]
    We prove part (i), and the proofs of parts (ii)--(iv) are analogous and hence omitted.
    First, we prove validity of $\mathcal S_{\mathtt{MTR}}(t)$.
    Consider any candidate (embedded) potential-outcomes process $\{X_i^s:s\in\mathcal T\}$ that is consistent with the observed pair $(X_i, D_i)$ and satisfy Assumption \ref{assumption:MTR-obj}, i.e., $X_i^{D_i}=X_i$ almost surely and $s_1\le s_2 \Rightarrow X_i^{s_1}\preceq X_i^{s_2}$ a.s.
    Define $\Omega_d := \{\omega\in\Omega: D_i(\omega) = d\}$, $\mathcal A_d := \{A \cap \Omega_d : A \in \mathcal A\}$, and $\mathbb P_d(A) = \P{A \mid D_i = d}$. We then define $\xi_{i}^s(\omega) := X_i^s(\omega), s\in\mathcal T$. By construction $\xi_{i}^s$ satisfies $\mathbb{P}_d(\xi_{i}^d \in A) = \mathbb{P}(X_i^d \in A \mid D_i=d) = \mathbb{P}(X_i \in A \mid D_i=d)$; $s_1\le s_2 \Rightarrow \xi_i^{s_1} \preceq \xi_i^{s_2}$, $\mathbb{P}_d$-a.s.; and $\mathbb{E}_{\mathbb P_d} [\|\xi_i^s\|_{L^2(\nu)}^2] = \mathbb{E}[\|X_i^s\|_{L^2(\nu)}^2 \mid D_i=d] < \infty$ by Assumption \ref{assumption:metric} and $\P{D_i=d}>0$.
    Here, define $f_d^s := \mathbb{E}[X_i^s \mid D_i=d]$. Then, $f_d^s = \mathbb{E}_{\mathbb P_d}[\xi_i^s]$, and 
    \begin{align*}
        \E{X_i^t}
        = \sum_{d\in\mathcal T}\P{D_i=d} f_d^t \in \mathcal S_{\mathtt{MTR}}(t),
    \end{align*}
    showing validity.
    
    Next, we prove sharpness. Take an arbitrary element of $\mathcal S_{\mathtt{MTR}}(t)$.
    Then there exists $(f_d^s)_{s\in\mathcal T} \in \mathcal F_d^{\mathtt{MTR}}$ for each $d\in\mathcal T$ such that the chosen arbitrary element is equal to $\sum_{d\in\mathcal T}\P{D_i=d} f_d^t$.
    By the definition of $\mathcal F_d^{\mathtt{MTR}}$, there exists a probability space $(\Omega_d, \mathcal A_d, \mathbb P_d)$ and a random element $(\xi_{i,d}^s)_{s\in\mathcal T}$ such that $\mathbb{P}_d[\xi_{i,d}^d \in A] = \P{X_i \in A \mid D_i = d}$, $s_1 \le s_2 \Rightarrow \xi_{i,d}^{s_1} \preceq \xi_{i,d}^{s_2}$ holds $\mathbb P_d$-a.s., $\mathbb{E}_{\mathbb{P}_d}[\xi_{i,d}^s]=f_d^s$ and $\mathbb{E}_{\mathbb{P}_d}[\|\xi_{i,d}^s\|_{L^2(\nu)}^2] < \infty$ for all $s\in\mathcal T$.
    We now construct a single probability space by taking the disjoint union of the cell-specific probability spaces. Define $\tilde \Omega := \bigsqcup_{d\in\mathcal T} (\{d\}\times \Omega_d)$.
    For each $B\subseteq \tilde \Omega$, define $B_d := \{\omega\in\Omega_d : (d,\omega)\in B\}$, and define $\tilde{\mathcal A} := \{B \subseteq \tilde\Omega : B_d \in \mathcal A_d \text{ for every }d\in\mathcal T\}$.
    Finally, define $\tilde{\mathbb{P}}(B) := \sum_{d\in\mathcal T} \P{D_i=d} \mathbb{P}_d(B_d)$ for $B\in\tilde{\mathcal A}$, and $\tilde D(d,\omega) := d$ and $\tilde X_i^s(d,\omega) := \xi_{i,d}^s(\omega), s\in\mathcal T$ for $(d,\omega)\in\{d\}\times\Omega_d$.
    Then, it follows that $\tilde{\mathbb P}(\tilde X_i^{\tilde D_i} \in A, \tilde D_i = d) = \mathbb{P}[D_i=d] \P{X_i \in A \mid D_i = d} = \P{X_i\in A, D_i = d}$, proving the observational equivalence.
    Furthermore, $s_1 \le s_2 \Rightarrow \tilde X_i^{s_1}=\xi_{i,d}^{s_1} \preceq \xi_{i,d}^{s_2} = \tilde X_i^{s_2}$, $\mathbb P_d$-a.s., implying the same conclusion follows $\tilde{\mathbb P}$-a.s. This shows that the constructed data-generating process satisfies MTR.
    Finally, $\mathbb{E}_{\tilde{\mathbb{P}}}[\tilde X_i^t] = \sum_d \P{D_i=d}\mathbb{E}_{\mathbb{P}_d}[\xi_{i,d}^t] = \sum_d \P{D_i=d} f_d^t$, which proves the attainability.
\end{proof}

\begin{proof}[Proof of Proposition \ref{prop:gen-simple}]
    Fix $\mathtt M \in \{\mathtt{MTR},\mathtt{MTS},\mathtt{MIV},\mathtt{MTR+MTS}\}$.
    We first show that $\mathcal B_{\mathtt M}(t)$ is a valid outer identified set for $\E{\Psi(Y_i(t))}$.
    Consider any candidate potential-outcomes process $\{X_i^s:s\in\mathcal T\}$ that is consistent with observed outcomes and satisfies the maintained monotonicity assumption $\mathtt M$. 
    There exists a set $\mathcal K^\ast\subseteq\mathcal K$ with $\nu(\mathcal K\backslash\mathcal K^\ast)=0$ such that, for every $k\in\mathcal K^\ast$, the scalar random variables $\{X_i^s(k):s\in\mathcal T\}$ are square-integrable and satisfy the scalar version of the maintained monotonicity restriction.
    Fix $k\in\mathcal K^\ast$. If $k\in\mathcal K_+$, the restriction has the same monotonicity direction as the scalar outcome. If $k\in\mathcal K_-$, the same scalar argument is applied to $-X_i^t(k)$. If $k\in\mathcal K_0$, no monotonicity restriction is imposed and only observational equivalence together with the support restriction is used.
    Hence, the scalar bounds of \cite{Manski:1997} and \cite{Manski_Pepper:2000} imply $\ell^{\mathtt M}_{t,k} \le \E{X_i^t(k)} \le u^{\mathtt M}_{t,k}$ for the fixed $k$.
    Then, for almost every $k\in\mathcal K$, $\ell^{\mathtt M}_{t,k} \le \E{X_i^t(k)} \le u^{\mathtt M}_{t,k}$.
    We therefore have $\E{X_i^t} = \E{\Psi(Y_i(t))} \in \mathcal B_{\mathtt M}(t)$. By Assumption \ref{assumption:metric}, $\E{\Psi(Y_i(t))} \in \Psi(\mathcal Y)$, so that it belongs to $\Psi(\mathcal Y)\cap \mathcal B_{\mathtt M}(t)$. 
    Finally, by the representation of the Fr\'echet mean under Assumption \ref{assumption:metric}, $\mathbb{E}_{\oplus}[Y_i(t)] = \Psi^{-1}\left(\E{\Psi(Y_i(t))}\right)$. Therefore every observationally equivalent counterfactual process satisfying $\mathtt M$ yields $\mathbb{E}_{\oplus}[Y_i(t)] \in \Psi^{-1}(\Psi(\mathcal Y)\cap \mathcal B_{\mathtt M}(t))$, proving the validity of $\Psi^{-1}(\Psi(\mathcal Y)\cap \mathcal B_{\mathtt M}(t))$. 
\end{proof}

\begin{proof}[Proof of Proposition \ref{prop:simple-sharp}]
    Fix $\mathtt M \in \{\mathtt{MTR},\mathtt{MTS},\mathtt{MIV},\mathtt{MTR+MTS}\}$.
    Since the validity is shown in Proposition \ref{prop:gen-simple}, it remains only to establish attainability. Let $\mathbb P_0$ denote the observed joint distribution of $(X_i, D_i, Z_i)$, where $Z_i$ is omitted when it is not relevant. Take any $b = (b_1,\ldots,b_p)^\top \in \mathcal B_{\mathtt{M}}(t)$.
    For every coordinate $k\in\{1,\ldots,p\}$, $b_k$ belongs to the corresponding scalar sharp identified interval. Hence, there exist a probability measure $\mathbb P_k$ and a scalar potential-outcome process $(U_{ik}^s)_{s \in \mathcal T}$ such that $\mathbb{E}_{\mathbb P_k}[U_{ik}^t] = b_k$, the process $(U_{ik}^s)_{s\in\mathcal T}$ satisfies the scalar restriction corresponding to $\mathtt M$, and observational equivalence in the sense that $\mathbb P_k[(U_{ik}^{D_i}, D_i, Z_i) \in A] = \mathbb P_0[(X_i(k), D_i, Z_i) \in A]$ for every Borel set $A$.

    We now combine these scalar models while preserving the observed joint distribution of the full vector $X_i$.
    Write $U_{ik} = (U_{ik}^s)_{s\in\mathcal T}$. Define the conditional joint distribution of $U_{i1},\ldots, U_{ip}$ by
    \begin{align*}
        \tilde{\mathbb P}\left[U_{i1} \in A_1,\ldots, U_{ip} \in A_p \mid X_i = x, D_i=d, Z_i = z\right]
        = 
        \prod_{k=1}^p {\mathbb P}_k\left[U_{ik} \in A_k \mid U_{ik}^{D_i} = x_k, D_i=d, Z_i = z\right].
    \end{align*}
    The marginal distribution of $(X_i, D_i, Z_i)$ under $\tilde{\mathbb P}$ is set equal to $\mathbb P_0$, that is,
    \begin{align*}
        &\tilde{\mathbb P}\left[(X_i, D_i, Z_i) \in A, U_{i1} \in A_1,\ldots, U_{ip} \in A_p\right]\\
        &\quad=
        \int_A \prod_{k=1}^p {\mathbb P}_k\left[U_{ik} \in A_k \mid U_{ik}^{D_i} = x_k, D_i=d, Z_i = z\right]\,d\mathbb P_0(x,d,z).
    \end{align*}
    Here, let $\mathbb P_{0,k}$ denote the marginal distribution of $(X_i(k), D_i, Z_i)$ under $\mathbb P_0$.
    Then
    \begin{align*}
        &\tilde{\mathbb P}\left[U_{ik} \in A_k, D_i \in B, Z_i \in C\right]\\
        &\quad=
        \int_{\mathbb{R} \times B \times C} {\mathbb P}_k\left[U_{ik} \in A_k \mid U_{ik}^{D_i} = x_k, D_i=d, Z_i = z\right]\,d\mathbb P_{0,k}(x,d,z).
    \end{align*}
    By the observational equivalence $\mathbb P_k[(U_{ik}^{D_i}, D_i, Z_i) \in A] = \mathbb P_0[(X_i(k), D_i, Z_i) \in A]$, $\mathbb P_{0,k}$ is also the distribution of $(U_{ik}^{D_i}, D_i, Z_i)$.
    Therefore,
    \begin{align*}
        &{\mathbb P}_k\left[U_{ik} \in A_k, D_i \in B, Z_i \in C\right]\\
        &\quad=
        \int_{\mathbb{R} \times B \times C} {\mathbb P}_k\left[U_{ik} \in A_k \mid U_{ik}^{D_i} = x_k, D_i=d, Z_i = z\right]\,d\mathbb P_{0,k}(x,d,z),
    \end{align*}
    and hence, $\tilde{\mathbb P}[U_{ik} \in A_k, D_i \in B, Z_i \in C] = {\mathbb P}_{k}[U_{ik} \in A_k, D_i \in B, Z_i \in C]$. This implies that 
    \begin{align*}
        \mathbb{E}_{\tilde{\mathbb P}}\left[U_{ik}^t\right] 
        = \mathbb{E}_{{\mathbb P}_k}\left[U_{ik}^t\right] = b_k.
    \end{align*}
    We next verify consistency with the observed vector. For every $k\in\{1,\ldots,p\}$,
    \begin{align*}
        \tilde{\mathbb P}\left[U_{ik}^{D_i} \neq X_i(k)\right] 
        &= \int \tilde{\mathbb P}\left[U_{ik}^{d} \neq x_k \mid X_i = x, D_i = d, Z_i = z\right]\,d \mathbb{P}_{0}(x,d,z) \\
        &= \int {\mathbb P}_k\left[U_{ik}^{d} \neq x_k \mid U_{ik}^{D_i} = x_k, D_i = d, Z_i = z\right]\,d \mathbb{P}_{0}(x,d,z) = 0,
    \end{align*}
    that is, $U_{ik}^{D_i} = X_i(k)$ holds $\tilde{\mathbb P}$-almost surely for every $k\in\{1,\ldots,p\}$.
    Define $U_i^s = (U_{i1}^s,\ldots U_{ip}^s)^\top$. Then we obtain $U_i^{D_i} = X_i$, $\tilde{\mathbb P}$-almost surely.
    Thus, the construction preserves the full joint distribution of the observed embedded outcome, not merely its coordinatewise marginal distributions.

    Because $\Psi(\mathcal Y) = \mathbb{R}^p$, any vector $U_i^s \in \mathbb R^p$ is a feasible embedded outcome.
    Define $\widetilde Y_i(s) = \Psi^{-1}(U_i^s)$ for $s\in\mathcal T$.
    Since $\Psi$ is now a bijective isometry, $\Psi^{-1}$ is continuous and hence measurable. Moreover,
    \begin{align*}
        \widetilde Y_i(D_i) = \Psi^{-1}(U_i^{D_i}) = \Psi^{-1}(X_i) = Y_i,
        \quad \widetilde{\mathbb P}\text{-almost surely}.
    \end{align*}
    Therefore, the constructed potential-outcome process is observationally equivalent to the original model, satisfies the maintained restriction, and attains $\mathbb E_{\widetilde{\mathbb P}}[\Psi(\widetilde Y_i(t))] = \mathbb E_{\widetilde{\mathbb P}}[U_i^t] = b$.
    Since $b\in\mathcal B_{\mathtt M}(t)$ was arbitrary, every point in $\mathcal B_{\mathtt M}(t)$ is attainable. Together with Proposition~\ref{prop:gen-simple}, this proves that $\mathcal B_{\mathtt M}(t)$ is the sharp identified set for $\E{\Psi(Y_i(t))}$.
    Since $\Psi$ is one-to-one, the sharp identified set for the \f mean is $\Psi^{-1}(\mathcal B_{\mathtt M}(t))$.
\end{proof}

\begin{proof}[Proof of Proposition \ref{prop:dist}]
    \textit{Proof of part (i)}.
    Write $Q_i(q) = Q_{Y_i}(q) = Q_{Y_i(D_i)}(q)$.  
    We first derive valid lower and upper bounds. 
    Take any candidate $(H_i^s)_{s\in \mathcal T}$ that is consistent with the observed data and MTR. Then $H_i^{D_i}=Q_i$ and $s_1\le s_2$ implies $H_i^{s_1}(q)\le H_i^{s_2}(q)$ almost every $q\in(0,1)$.
    The latter implies that $H_i^{s_1}(q)\le H_i^{s_2}(q)$ for every $q\in(0,1)$. 
    This is because, if it failed at some $q_0$, left continuity would imply that it also failed on an interval immediately to the left of $q_0$, contradicting the almost-everywhere inequality.
    Note that $H_i^t(q)\ge Q_i(q)1\{D_i\le t\}+K_0 1\{D_i>t\}$ and $H_i^t(q)\le Q_i(q)1\{D_i\ge t\}+K_1 1\{D_i<t\}$.
    Taking expectations, we obtain the pointwise lower and upper bounds provided in the statement of Proposition \ref{prop:dist}.

    It remains to verify that the two bounding functions are themselves attainable elements of the sharp identified set.
    Define the counterfactual quantile paths by
    \begin{align*}
        \underline H_i^s(q)
        &\coloneqq
        Q_i(q)\mathbf 1\{D_i\le s\}
        +
        K_0\mathbf 1\{D_i>s\},\quad s\in\mathcal T\\
        \overline H_i^s(q)
        &\coloneqq
        Q_i(q)\mathbf 1\{D_i\ge s\}
        +
        K_1\mathbf 1\{D_i<s\},\quad s\in\mathcal T.
    \end{align*}
    For every $s\in\mathcal T$, both $\underline H_i^s$ and $\overline H_i^s$ belong to $\mathcal Q_{[K_0,K_1]}$.
    Moreover, $\underline H_i^{D_i}=\overline H_i^{D_i}=Q_i$.
    MTR also holds.
    Finally, $\mathbb{E}[\underline H_i^t(q)] = \underline{Q}^{\mathtt{MTR}}(\cdot;t)$ and $\mathbb{E}[\overline H_i^t(q)] = \overline{Q}^{\mathtt{MTR}}(\cdot;t)$.
    Thus the lower and upper bounding functions are themselves attainable, showing the uniform sharpness.\hfill $\square$

    \noindent \textit{Proof of part (ii)}.
    The validity can be shown similarly to the argument above. We thus show attainability.
    For each $d\in\mathcal T$, define $\overline Q_d(q) = \E{Q_i(q)\mid D_i=d} \in\mathcal Q_{[K_0,K_1]}$.
    For the lower bound, define, for every $s\in\mathcal T$,
    \begin{align*}
        \xi_{i,L}^s(q)
        =
        \begin{cases}
            K_0, & D_i<s,\\
            Q_i(q), & D_i=s,\\
            \overline Q_s(q), & D_i>s.
        \end{cases}
    \end{align*}
    For each $s\in\mathcal T$, $\xi_{i,L}^s$ belongs to $\mathcal Q_{[K_0,K_1]}$ almost surely and satisfies $\xi_{i,L}^{D_i} = Q_i$ almost surely. Moreover,
    \begin{align*}
        \E{\xi_{i,L}^s(q)\mid D_i=d}
        =
        \begin{cases}
            K_0, & d<s,\\
            \overline Q_s(q), & d\ge s.
        \end{cases}
    \end{align*}
    Thus, for every $s\in\mathcal T$, the conditional mean quantile function is weakly increasing in $d$, so that $(\xi_{i,L}^s)_{s\in\mathcal T}$ satisfies Assumption \ref{assumption:D-MTS}. 
    Its target mean is
    \begin{align*}
        \E{\xi_{i,L}^t(q)}
        =
        K_0\P{D_i<t}+\overline Q_t(q)\P{D_i\ge t}
        =
        \underline Q_{\mathtt{MTS}}(q;t).
    \end{align*}
    Therefore, $\underline Q_{\mathtt{MTS}}(\cdot;t)$ is attainable.
    The upper bound can be treated similarly.
    \hfill $\square$

    \noindent \textit{Proof of part (iii)}. 
    We begin by rewritng the sharp identified set $\mathcal S_{\mathtt{MIV}}$ in Lemma \ref{lem:obj}(iii).
    Let $\mathcal F_{\mathtt{MIV}}$ denote the feasible collection of $f=\{f_{d,z}^s:s\in\mathcal T,\ d\in\mathcal T,\ z\in\mathcal Z\}$, namely,
    \begin{align*}
        \mathcal F_{\mathtt{MIV}} =
        \left\{
            f:
            \begin{array}{l} 
                (f_{d,z}^s)_{s\in\mathcal T}\in\mathcal F_{d,z}^0 \text{ for all }(d,z)\in\mathcal T\times\mathcal Z,\\
                z_1\le z_2 \Rightarrow m_{z_1}^s(f)\preceq m_{z_2}^s(f) \text{ for all }s\in\mathcal T 
            \end{array}
        \right\},
    \end{align*}
    where $m_z^s(f) = \sum_{d\in\mathcal T}\P{D_i=d\mid Z_i=z}f_{d,z}^s$.
    Take any $f\in\mathcal F_{\mathtt{MIV}}$.
    For each $z\in\mathcal Z$, define 
    \begin{align*}
        v_z^f(q) = 
        \frac{\sum_{d\neq t}\P{D_i=d\mid Z_i=z}f_{d,z}^t(q)}{\P{D_i\neq t\mid Z_i=z}}.
    \end{align*}
    Since $\mathcal Q_{[K_0, K_1]}$ is convex, $v_z^f\in\mathcal Q_{[K_0,K_1]}$ for every $z\in\mathcal Z$, because it is a convex conbination of the quantile functions $f_{d,z}^t$.
    Furthermore,
    \begin{align*}
        [m_z^t(f)](q)
        &=
        \P{D_i=t\mid Z_i=z}f_{t,z}^t(q)
        +
        \sum_{d\neq t}
        \P{D_i=d\mid Z_i=z}f_{d,z}^t(q) \\
        &=
        \E{Q_i(q)\mathbf{1}\{D_i=t\}\mid Z_i=z}
        +
        \P{D_i\neq t\mid Z_i=z}v_z^f(q) = g_z(q;t,v^f),
    \end{align*}
    which shows, along with $f\in\mathcal F_{\mathtt{MIV}}$ and the same argument upgrading the a.e. order to the pointwise order used above, that $v^f\in\mathcal V_t$.
    It therefore follows that $\mathcal{S}_{\mathtt{MIV}} \subseteq \{\sum_{z\in\mathcal Z} \P{Z_i =z} g_z(\cdot; t,v): v\in\mathcal V_t\}$.
    
    Conversely, take any $v\in\mathcal V_t$. Let $f^\circ\in\mathcal F_{\mathtt{MIV}}$.
    Define $f^v$ by
    \begin{align*}
        f_{d,z}^{s,v}
        =
        \begin{cases}
            f_{d,z}^{s,\circ},
            & s\neq t,\\
            f_{t,z}^{t,\circ},
            & s=t,\ d=t,\\
            v_z,
            & s=t,\ d\neq t.
        \end{cases}
    \end{align*}
    For every cell $(d,z)$, the observed coordinate is preserved, so that $(f_{d,z}^{s,v})_{s\in\mathcal T} \in \mathcal F_{d,z}^0$ for every $(d,z)\in\mathcal T\times\mathcal Z$.
    For $s\neq t$, the MIV restrictions are inherited from $f^\circ$.
    For $s=t$, we can show that $[m_z^t(f^v)](q) = g_z(q;t,v)$, which proves, along with $v\in\mathcal V_t$, that $f^v\in\mathcal F_{\mathtt{MIV}}$.
    Therefore, we obtain the reversed inclusion, and in turn have that
    \begin{align*}
        \mathcal S_{\mathtt{MIV}}(t) 
        &=
        \left\{ \sum_{z\in\mathcal Z}\P{Z_i=z}\, g_z(\cdot\,;t,v):
        \begin{array}{l} 
            v \in \mathcal{V}_t
        \end{array}
        \right\}.
    \end{align*}
    
    We next show that $\ell^t,u^t\in\mathcal V_t$.
    Fix $z\in\mathcal Z$. 
    Since every $v_z\in\mathcal Q_{[K_0,K_1]}$ is nondecreasing, $\ell_z^t$ is nondecreasing. Also, $K_0\le\ell_z^t(q)\le K_1$ for all $q\in(0,1)$. For $z_1\le z_2$, the definition of $\mathcal V_t$ implies that $g_{z_1}(q;t,v)\le g_{z_2}(q;t,v)$ for every $v\in\mathcal V_t$ and $q\in(0,1)$. Since $\P{D_i\neq t\mid Z_i=z}\ge0$, we have $g_{z_1}(q;t,\ell^t) = \inf_{v\in\mathcal V_t}g_{z_1}(q;t,v) \le \inf_{v\in\mathcal V_t}g_{z_2}(q;t,v) = g_{z_2}(q;t,\ell^t)$.
    Define $\ell_z^{t,-}(q)\coloneqq\lim_{r\uparrow q}\ell_z^t(r)$ and $\ell^{t,-}\coloneqq(\ell_z^{t,-})_{z\in\mathcal Z}$. Since $\ell_z^t$ is nondecreasing and bounded, $\ell_z^{t,-}$ is nondecreasing, bounded, and left continuous. Taking $r\uparrow q$ in the preceding inequality yields $g_{z_1}(q;t,\ell^{t,-})\le g_{z_2}(q;t,\ell^{t,-})$ for every $z_1\le z_2$ and $q\in(0,1)$. Hence, $\ell^{t,-}\in\mathcal V_t$. Therefore, $\ell_z^t(q)\le\ell_z^{t,-}(q)$ by the definition of $\ell_z^t$. On the other hand, monotonicity of $\ell_z^t$ implies $\ell_z^{t,-}(q)\le\ell_z^t(q)$. Thus, $\ell_z^t(q)=\ell_z^{t,-}(q)$ for every $q\in(0,1)$. It follows that $\ell_z^t$ is left continuous, and hence $\ell^t\in\mathcal V_t$.
    Similarly, $u_z^t \in\mathcal V_t$.
    This shows that $\underline{Q}_{\mathtt{MIV}}(q;t)$ and $\overline{Q}_{\mathtt{MIV}}(q;t)$ are elements of $\mathcal S_{\mathtt{MIV}}$.
    
    Finally, we show that $\underline{Q}_{\mathtt{MIV}}(q;t)$ and $\overline{Q}_{\mathtt{MIV}}(q;t)$ are the envelopes of $\mathcal S_{\mathtt{MIV}}$.
    Recall that, for every $v\in\mathcal V_t$, $\ell_z^t(q)\le v_z(q)\le u_z^t(q)$ for all $z\in\mathcal Z$ and $q\in(0,1)$. 
    Hence, $g_z(q;t,\ell^t) \leq g_z(q;t,v) \leq g_z(q;t,u^t)$.
    This in turn implies that
    \begin{align*}
        \underline{Q}_{\mathtt{MIV}}(q;t)
        &=
        \sum_{z\in\mathcal Z} g_z(q;t,\ell^t)\P{Z_i=z}\\
        &\leq 
        \sum_{z\in\mathcal Z} g_z(q;t,v)\P{Z_i=z}
        \leq 
        \sum_{z\in\mathcal Z} g_z(q;t,u^t)\P{Z_i=z}
        =
        \overline{Q}_{\mathtt{MIV}}(q;t).
    \end{align*}
    Thus, $\underline{Q}_{\mathtt{MIV}}$ and $\overline{Q}_{\mathtt{MIV}}$ are the uniformly sharp lower and upper bounds, respectively.\hfill $\square$

    \noindent \textit{Proof of part (iv)}.
    By Lemma \ref{lem:obj}(iv), 
    \begin{align*}
        \mathcal S_{\mathtt{MTR+MTS}}(t) = 
        \left\{ \sum_{d\in\mathcal T}\P{D_i=d}\, f_d^t: 
        \begin{array}{l} 
            (f_d^s)_{s\in\mathcal T} \in \mathcal F_d^{\mathtt{MTR}} \text{ for all }d\in\mathcal T,\\
            d_1\le d_2 \Rightarrow f_{d_1}^s(q)\le f_{d_2}^s(q) \text{ for all }s\in\mathcal T, \text{a.e. }q\in(0,1)
        \end{array} 
        \right\}.
    \end{align*}
    The same a.e.-to-pointwise argument used above implies the inequality for every $q\in(0,1)$.
    Take any $h \in \mathcal S_{\mathtt{MTR+MTS}}(t)$, and let $(f_d^s)_{d,s}$ be a feasible collection such that $h=\sum_{d\in\mathcal T}\P{D_i=d}\, f_d^t$. 
    Then, $f_d^d(q) = \E{Q_i(q)\mid D_i=d}$; for $d<t$, MTR imples $f_d^t(q)\ge f_d^d(q)=\E{Q_i(q)\mid D_i=d}$; and for $d \geq t$, MTS implies $f_d^t(q)\ge f_t^t(q)=\E{Q_i(q)\mid D_i=t}$. Hence we obtain that
    \begin{align*}
        \sum_{d\in\mathcal T} \P{D_i=d}f_d^t(q)
        \geq 
        \sum_{d<t} \E{Q_i(q)\mid D_i=d} \P{D_i=d} + \E{Q_i(q)\mid D_i=t} \P{D_i\geq t},
    \end{align*}
    which shows that the lower bound function is a valid bound. The upper bound can be treated similarly.
        
    It remains to prove attainability.
    Define 
    \begin{align*}
        f_d^{s,L}(q) = \E{Q_i(q)\mid D_i=d\wedge s}.
    \end{align*}
    We first show that $(f_d^{s,L})_{s\in\mathcal T} \in \mathcal F_d^{\mathtt{MTR}}$ for any $d\in\mathcal T$.
    Note that, under MTR and MTS,
    \begin{align*}
        &\E{Q_{i}(q)\mid D_i=d_1} = \E{Q_{Y_i(d_1)}(q)\mid D_i=d_1}\\
        &\quad\le \E{Q_{Y_i(d_2)}(q)\mid D_i=d_1}
        \le \E{Q_{Y_i(d_2)}(q)\mid D_i=d_2} 
        = \E{Q_{i}(q)\mid D_i=d_2}
    \end{align*}
    for $d_1 \le d_2$.
    Let $\mathcal T_{\le d}=\{s\in\mathcal T: s\le d\} = \{s_1<\cdots<s_r=d\}$.
    Lemma \ref{lem:mtr} below shows that there exists an individual-level MTR path for $s\le d$
    \begin{align*}
        \P{\tilde{Q}_i^{s_1}(q) \leq \cdots \leq \tilde{Q}_i^{d}(q) = Q_i(q)\text{ a.e. } q\in(0,1) \mid D_i=d}=1,
    \end{align*}
    and $q\mapsto \tilde{Q}_i^s(q)$ belongs to $\mathcal Q_{[K_0,K_1]}$ almost surely conditional on $D_i=d$ for every $s\in\mathcal T_{\le d}$. 
    Then the inequalities holds for every $q\in(0,1)$ by the same a.e.-to-pointwise argument.
    For $s\le d$, Lemma \ref{lem:mtr} shows that the conditional mean satisfies $\mathbb{E}[\tilde{Q}_i^{s}(q)\mid D_i=d] = \E{Q_i(q)\mid D_i=s}$.
    Hence we obtain $\mathbb{E}[\tilde{Q}_i^{s}(q)\mid D_i=d] = \E{Q_i(q)\mid D_i=d\wedge s} = f_d^{s,L}(q)$ for $s\le d$.
    Set the potential quantile function equal to $Q_i$ for all $s> d$.
    This yields
    \begin{align*}
        (f_d^{s,L})_{s\in\mathcal T} \in \mathcal F_d^{\mathtt{MTR}}.
    \end{align*}
    
    Next, we verify MTS.
    For each fixed $s$, $d \mapsto f_d^{s,L}(q) = \E{Q_i(q)\mid D_i=d\wedge s}$ is nondecreasing, showing $(f_d^{s,L})_{d,s\in\mathcal T}$ satisfy the MTS.
    
    Finally, the constructed $(f_d^{s,L})_{d,s\in\mathcal T}$ implies
    \begin{align*}
        f_d^{t,L}(q)
        =
        \begin{cases}
            \E{Q_{i}(q)\mid D_i=d}, & (d<t),\\
            \E{Q_{i}(q)\mid D_i=t}, & (d\ge t),
        \end{cases}
    \end{align*}
    and therefore
    \begin{align*}
        \sum_{d\in\mathcal T}
        \P{D_i=d}f_d^{t,L}(q)
        &=
        \sum_{d<t}\P{D_i=d} \E{Q_{i}(q)\mid D_i=d}
        +
        \P{D_i\ge t} \E{Q_{i}(q)\mid D_i=t}\\
        &=
        \underline{Q}_{\mathtt{MTR+MTS}}(q;t).
    \end{align*}
    Thus, $\underline{Q}_{\mathtt{MTR+MTS}}(\cdot;t) \in \mathcal S_{\mathtt{MTR+MTS}}$, i.e., it is the uniformly sharp lower bound. The upper bound is analogous.
\end{proof}

\begin{lemma}\label{lem:mtr}
    Let $X:\Omega\times(0,1)\to[K_0,K_1]$ be an $\mathcal A\otimes\mathcal B((0,1))$-measurable map such that $q\mapsto X(\omega,q)\in\mathcal Q_{[K_0,K_1]}$ for $\mathbb P$-almost every $\omega$. 
    Let 
    \begin{align*}
        h_1(q)\le\cdots\le h_r(q)=\E{X(q)},\quad h_j\in\mathcal Q_{[K_0,K_1]},
    \end{align*}
    for every $q\in(0,1)$.
    Then there exist $\mathcal A\otimes\mathcal B((0,1))$-measurable maps $H_j:\Omega\times(0,1)\to[K_0,K_1]$, $j=1,\ldots,r$ such that $q\mapsto H_j(\omega,q)\in\mathcal Q_{[K_0,K_1]}$ for $\mathbb P$-almost every $\omega$,
    \begin{align*}
    \P{H_1(q)\le\cdots\le H_r(q)=X(q)\text{ for almost every }q\in(0,1)}=1,
    \end{align*}
    and $\E{H_j(q)}=h_j(q)$ for almost every $q\in(0,1)$ and $j=1,\ldots,r$.
\end{lemma}
\begin{proof}[Proof of Lemma \ref{lem:mtr}]
    \textit{Step 1 (Preliminaries)}. Let $\mathbb L^2 = L^2(\Omega\times(0,1), \mathcal A\otimes\mathcal B((0,1)), \mathbb P\otimes\lambda)$. Since $X$ is jointly measurable and bounded, $X\in\mathbb L^2$.

    By translating all functions by $K_0$, it suffices to consider the case $K_0=0$. Write $B:=K_1-K_0$.
    For $n \in \mathbb N$, let $\Delta_n= B2^{-n}$ and define 
    \begin{align*}
        \mathcal C_n = 
        \left\{
            x=(x_0,\ldots,x_{2^n})\in[0,B]^{2^n+1}:
            x_0\le\cdots\le x_{2^n}
        \right\}.
    \end{align*}
    We first see a decomposition property of $\mathcal C_n$: if $x,y,z\in\mathcal C_n$ satisfy $z\le x+y$, then there exists $a,b\in\mathcal C_n$ such that $z=a+b$, $a\le x$, and $b\le y$, where inequalities between vectors are understood coordinatewise. We use the same notation below.
    Indeed, define $a_j = \max_{0\le k\le j}\{z_k-y_k\}_+$ and $b_j = z_j-a_j$ for $j=0,\ldots,2^n$. 
    Then, $a \in \mathcal C_n$ by construction.
    For every $k\le j$, $z_k-y_k\le x_k\le x_j$ and hence $a_j\le x_j$. 
    Furthermore, note that $\{z_k-y_k\}_+\le z_k\le z_j$ for $k\le j$, so that $0\le a_j\le z_j$, implying that $b_j\ge0$. Since $a_j\ge\{z_j-y_j\}_+$, we also have $b_j\le y_j$ by the definition of $b_j$.
    It remains to show that $b$ is nondecreasing.
    It suffices to see $a_j-a_{j-1} \leq z_j-z_{j-1}$.
    This inequality is immediate if $a_j=a_{j-1}$. Below we treat the case where $a_j>a_{j-1}$. In this case, $a_j = \{z_j - y_j\}_+ = z_j - y_j$, where the last equality uses $a_j>0$.
    If $z_{j-1}\ge y_{j-1}$, 
    \begin{align*}
        a_j-a_{j-1} \le (z_j-y_j)-(z_{j-1}-y_{j-1})
        = (z_j-z_{j-1})-(y_j-y_{j-1}) \le z_j-z_{j-1}.
    \end{align*}
    If $z_{j-1}< y_{j-1}$,
    \begin{align*}
        a_j-a_{j-1} \le a_j = z_j-y_j
        \le z_j-y_{j-1} < z_j-z_{j-1}.
    \end{align*}
    Hence, $a,b\in\mathcal C_n$.

    For $x\in\mathcal C_n$, define 
    \begin{align*}
        \mathcal{D}_{r,n}(x) 
        =
        \left\{
            (u_1,\ldots,u_r)\in\mathcal C_n^r:
            u_1\le\cdots\le u_r=x
        \right\}.
    \end{align*}
    We will see that, as long as $x,y,x+y\in \mathcal C_n$,
    \begin{align*}
        \mathcal D_{r,n}(x+y)&=\mathcal D_{r,n}(x)\oplus\mathcal D_{r,n}(y)\\
        &\coloneqq 
        \left\{
            (u_1+v_1, \ldots, u_r+v_r) : 
            (u_1, \ldots, u_r) \in \mathcal D_{r,n}(x)\text{ and } 
            (v_1, \ldots, v_r) \in \mathcal D_{r,n}(y)
        \right\}.
    \end{align*}
    $\mathcal D_{r,n}(x)\oplus\mathcal D_{r,n}(y) \subseteq \mathcal D_{r,n}(x+y)$ is immediate.
    We prove $\mathcal D_{r,n}(x+y)\subseteq \mathcal D_{r,n}(x)\oplus\mathcal D_{r,n}(y)$ by induction on $r$. The result is immediate for $r=1$.
    Suppose it holds for $r$ and take $(z_1,\ldots,z_{r+1}) \in \mathcal D_{r+1,n}(x+y)$.
    Then, $z_1 \le \cdots \le z_r \le z_{r+1} = x+y$, and in particular $z_r \le x+y$. We have shown that there exist $a_r,b_r\in\mathcal C_n$ such that $z_r=a_r+b_r$, $a_r\le x$, and $b_r\le y$.
    Now, noting that $(z_1,\ldots,z_r) \in \mathcal D_{r,n}(z_r) =  \mathcal D_{r,n}(a_r+b_r)$, it follows that $(z_1,\ldots,z_r) \in \mathcal D_{r,n}(a_r) \oplus \mathcal D_{r,n}(b_r)$ by hypothesis.
    Thus, there are some $(a_1,\ldots,a_r) \in \mathcal D_{r,n}(a_r)$ and $(b_1,\ldots,b_r) \in \mathcal D_{r,n}(b_r)$ such that $z_j = a_j + b_j$ for $j=1,\ldots,r$.
    If we set $a_{r+1} = x$ and $b_{r+1} = y$, $(a_1,\ldots,a_{r+1}) \in \mathcal D_{r+1,n}(x)$, $(b_1,\ldots,b_{r+1}) \in \mathcal D_{r+1,n}(y)$, and $z_{r+1} = x+y = a_{r+1} + b_{r+1}$. This proves $(z_1,\ldots,z_{r+1}) \in \mathcal D_{r+1,n}(x)\oplus\mathcal D_{r+1,n}(y)$.

    \textit{Step 2 (Discrete approximations)}. Extend $X(\omega,\cdot)$ to $[0,1)$ by setting $\overline{X}(\omega,0)=0$ and $\overline{X}(\omega,q) = X(\omega,q)$. The map $\overline X$ is jointly measurable on $\Omega\times[0,1)$. Moreover, for $\mathbb P$-almost every $\omega$, the function $q\mapsto\overline X(\omega,q)$ is nondecreasing and takes values in $[0,B]$.

    For a bounded function $f:[0,1)\to[0,B]$, define
    \begin{align*}
        (S_nf)(q)
        =
        f\left(\frac{\lceil 2^nq\rceil-1}{2^n}\right),
        \quad q\in(0,1),
    \end{align*}
    and define
    \begin{align*}
        \rho_n(x)
        =
        \Delta_n \left\lfloor\frac{x}{\Delta_n}\right\rfloor,
        \quad x\in[0,B].
    \end{align*}
    Set $X_n(\omega,q) = \rho_n((S_n\overline X)(\omega,q))$, which is $\mathcal A \otimes\mathcal B((0,1))$-measurable.

    For $\mathbb P$-almost every $\omega$, $q\mapsto X_n(\omega,q)$ is a nondecreasing left-continuous step function with deterministic breakpoints. Therefore, $X_n(\omega,\cdot)\in\mathcal Q_{[0,B]}$ for $\mathbb P$-almost every $\omega$.

    For $\mathbb P$-almost every $\omega$ and every $q\in(0,1)$, $0 \le X_n(\omega,q) \le (S_n\overline X)(\omega,q) \le X(\omega,q)$. Thus, $0\le X_n\le X$ holds $(\mathbb P\otimes\lambda)$-almost everywhere.

    We next show that $X_n$ converges to $X$ in $\mathbb L^2$.
    Note that
    \begin{align*}
        \|X_n-X\|_{\mathbb L^2}
        \le
        \|X_n-S_n\overline X\|_{\mathbb L^2}
        +
        \|S_n\overline X-X\|_{\mathbb L^2}.
    \end{align*}
    By construction, $|X_n(\omega,q) - (S_n\overline X)(\omega,q)| \le \Delta_n$, and hence, the first term converges to zero.
    For $\mathbb P$-almost every $\omega$, left continuity of $X(\omega,\cdot)$ implies $(S_n\overline X)(\omega,q) - X(\omega,q) \to 0$ for every $q\in(0,1)$.
    Then, by the dominated convergence theorem, $\|S_n\overline X-X\|_{\mathbb L^2} \to 0$.
    Consequently, $\|X_n-X\|_{\mathbb L^2} \to 0$.

    \textit{Step 3 (Approximation of $h$)}.
    Write
    \begin{align*}
        \mu(q) = \int_\Omega X(\omega,q)\,d\mathbb P(\omega),\quad 
        \mu_n(q) = \int_\Omega X_n(\omega,q)\,d\mathbb P(\omega)
    \end{align*}
    for every $q\in(0,1)$. As in Step 2, extend $\mu$ and each $h_j$ to $[0,1)$ by setting $\mu(0) = 0$ and $h_j(0) = 0$.
    Define, for $j=1,\ldots,r$,
    \begin{align*}
        h_{j,n} = (S_nh_j)\wedge\mu_n.
    \end{align*}
    Since $X_n\le S_n\overline X$, we have
    \begin{align*}
        \mu_n(q) \le \int_\Omega (S_n\overline X)(\omega,q)\,d\mathbb P(\omega)
        = (S_n \mu)(q) = (S_nh_r)(q)
    \end{align*}
    Therefore, $h_{r,n}(q) = ((S_nh_r)\wedge\mu_n)(q) = \mu_n(q)$ and $h_{1,n}(q) \le \cdots \le h_{r,n}(q) = \mu_n(q)$ for every $q\in(0,1)$.

    Define $\iota_n(f) = (c_0,\ldots,c_{2^n-1},c_{2^n-1})$, where $f(q) = c_{k-1}$ for $q \in ((k-1)/2^n, k/2^n]$ and $k=1,\ldots,2^n$.
    Define $\tilde h_{j,n} = \iota_n(h_{j,n})$ for $j=1,\ldots,r$ and $\tilde \mu_n = \iota_n(\mu_n)$.
    Then we have $(\tilde h_{1,n},\ldots,\tilde h_{r,n}) \in \mathcal D_{r,n}(\tilde \mu_n)$.

    We next show that $h_{j,n} \to h_{j}$ in $L^2((0,1),\lambda)$.
    By Jensen's inequality,
    \begin{align*}
        \|\mu_n-\mu\|_{L^2((0,1),\lambda)}^2 = 
        \int_0^1 \left|\E{X_n(q)-X(q)}\right|^2\,dq
        \leq \int_0^1 \E{\left|X_n(q)-X(q)\right|^2}\,dq
        = \|X_n-X\|_{\mathbb L^2}^2,
    \end{align*}
    where the last expression converges to zero by Step 2.
    Further, left continuity of $h_j$ implies $(S_nh_j)(q) \to h_j(q)$ for $q\in(0,1)$, and hence, the dominated convergence theorem gives $\|S_nh_j-h_j\|_{L^2((0,1),\lambda)} \to 0$.
    Using $h_j=h_j\wedge\mu$ and $|a\wedge b-c\wedge d| \le |a-c|+|b-d|$, we obtain
    \begin{align*}
        \|h_{j,n}-h_j\|_{L^2((0,1),\lambda)}
        \le
        \|S_nh_j-h_j\|_{L^2((0,1),\lambda)} + \|\mu_n-\mu\|_{L^2((0,1),\lambda)}
        \to 0.
    \end{align*}

    \textit{Step 4 (Constructing discrete, monotone random paths)}.
    By construction in Step 2, $X_n$ takes only finitely many deterministic values.
    Then there exist $L_n<\infty$ and distinct vectors $x_n^{(1)},\ldots,x_n^{(L_n)}\in\mathcal C_n$ such that $\mathbb{P}[\iota_n(X_n)\in\{x_n^{(1)},\ldots,x_n^{(L_n)}\}] = 1$.
    Define
    \begin{align*}
        A_{n,\ell}
        =
        \left\{\omega\in\Omega:
        \iota_n(X_n)(\omega)=x_n^{(\ell)}
        \right\},\quad
        p_{n,\ell} =\P{A_{n,\ell}},
        \quad
        \ell=1,\ldots,L_n.
    \end{align*}
    Then $p_{n,\ell}>0$ for every $\ell$, and $\tilde \mu_n = \E{\iota_n(X_n)} = \sum_{\ell=1}^{L_n} p_{n,\ell}x_n^{(\ell)}$.
    Since, $(\tilde h_{1,n},\ldots,\tilde h_{r,n}) \in \mathcal D_{r,n}(\tilde\mu_n)$, the decomposition property established in Step 1 yields 
    \begin{align*}
    \mathcal D_{r,n}(\tilde \mu_n)
    =
    \mathcal D_{r,n}\left(\sum_{\ell=1}^{L_n}p_{n,\ell}x_n^{(\ell)}\right)
    =
    \bigoplus_{\ell=1}^{L_n} p_{n,\ell}\mathcal D_{r,n}\left(x_n^{(\ell)}\right).
    \end{align*}
    Hence, for every $\ell=1,\ldots,L_n$, there exists $(u_{1,n}^{(\ell)},\ldots,u_{r,n}^{(\ell)}) \in \mathcal D_{r,n}(x_n^{(\ell)})$ such that $(\tilde h_{1,n},\ldots,\tilde h_{r,n}) = \sum_{\ell=1}^{L_n} p_{n,\ell} (u_{1,n}^{(\ell)},\ldots,u_{r,n}^{(\ell)})$.

    Let $u_{j,n}^{(\ell)} = (u_{j,n,0}^{(\ell)}, \ldots, u_{j,n,2^n}^{(\ell)})$.
    For $j=1,\ldots,r$, define 
    \begin{align*}
        H_{j,n}(\omega,q)
        =
        \sum_{\ell=1}^{L_n}\sum_{k=1}^{2^n}\mathbf{1}\{\omega \in A_{n,\ell}\} 
        \mathbf{1}\left\{q\in \left(\frac{k-1}{2^n},\frac{k}{2^n}\right]\right\}
        u_{j,n,k-1}^{(\ell)},
    \end{align*}
    for $(\omega,q)\in\Omega\times(0,1)$. Then each $H_{j,n}$ is $\mathcal A \otimes\mathcal B((0,1))$-measurable. For every $\ell$, $u_{1,n}^{(\ell)} \le \cdots \le u_{r,n}^{(\ell)} = x_n^{(\ell)}$. Hence, $0 \le H_{1,n} \le \cdots \le H_{r,n} = X_n$ holds $(\mathbb P\otimes\lambda)$-almost everywhere.
    Moreover, each $H_{j,n}(\omega,\cdot)$ is a nondecreasing, left-continuous step function taking values in $[0,B]$, so that $H_{j,n}(\omega,\cdot) \in \mathcal Q_{[0,B]}$ for $\mathbb P$-almost every $\omega$. Finally, for every $j=1,\ldots,r$ and every $q\in ((k-1)/2^n, k/2^n]$
    \begin{align*}
    \E{H_{j,n}(q)}
    =
    \sum_{\ell=1}^{L_n} p_{n,\ell}u_{j,n,k-1}^{(\ell)}
    =
    h_{j,n}(q).
    \end{align*}

    \textit{Step 5 (Weak limit)}.
    Let $\mathcal M = \{f\in L^2((0,1),\lambda): f\text{ admits a nondecreasing representative}\}$. Then the set $\mathcal M$ is a closed and convex in $L^2((0,1),\lambda)$.
    Define
    \begin{align*}
        \mathcal A_X
        =
        \left\{(W_1,\ldots,W_r)\in (\mathbb L^2)^r:
        \begin{array}{l}
        W_j(\omega,\cdot)\in\mathcal M \text{ for }\mathbb P\text{-almost every }\omega, \quad j=1,\ldots,r,\\
        0\le W_1\le\cdots\le W_r\le X
        \quad (\mathbb P\otimes\lambda)\text{-almost everywhere}
        \end{array}
        \right\}.
    \end{align*}
    Then $\mathcal A_X$ is convex and norm closed in $(\mathbb L^2)^r$.
    Indeed, since $\mathcal M$ is convex, the first restriction is convex. The order restrictions also define a closed convex subset of $(\mathbb L^2)^r$. Hence, $\mathcal A_X$ is convex.
    To prove norm closedness of $\mathcal A_X$, let $(W_{1,m},\ldots,W_{r,m})\in\mathcal A_X$ and suppose that $(W_{1,m},\ldots,W_{r,m}) \to (W_1,\ldots,W_r)$ in $(\mathbb L^2)^r$. Since the order restrictions are norm closed, $0\le W_1\le\cdots\le W_r\le X$ holds $(\mathbb P\otimes\lambda)$-almost everywhere.    
    It remains to verify the first restriction. Define
    \begin{align*}
        D_m(\omega)^2
        =
        \sum_{j=1}^r \|W_{j,m}(\omega,\cdot)-W_j(\omega,\cdot)\|_{L^2((0,1),\lambda)}^2.
    \end{align*}
    By Fubini's theorem,
    \begin{align*}
        \int_\Omega D_m(\omega)^2\,d\mathbb P(\omega)
        &=
        \sum_{j=1}^r\|W_{j,m}-W_j\|_{\mathbb L^2}^2
        \to 0.
    \end{align*}
    Thus, there exists a subsequence $D_{m_k}(\omega)\to 0$ for $\mathbb P$-almost every $\omega$.
    For every $j=1,\ldots,r$, this implies $W_{j,m_k}(\omega,\cdot) \to W_j(\omega,\cdot)$ in $L^2((0,1),\lambda)$ for $\mathbb P$-almost every $\omega$.
    Moreover, since each $(W_{1,m_k},\ldots,W_{r,m_k})$ belongs to $\mathcal A_X$, we have $W_{j,m_k}(\omega,\cdot)\in\mathcal M$ for every $j,k$ and for $\mathbb P$-almost every $\omega$.
    Taking the intersection of these countably many full-measure sets,
    we may assume that this holds simultaneously for all $j$ and $k$.
    Since $\mathcal M$ is closed in $L^2((0,1),\lambda)$, it follows that $W_j(\omega,\cdot)\in\mathcal M$ for $\mathbb P$-almost every $\omega$ and every $j=1,\ldots,r$.
    Therefore, $\mathcal A_X$ is norm closed in $(\mathbb L^2)^r$.
    
    Since $(\mathbb L^2)^r$ is a Hilbert space, every norm closed convex
    subset is weakly closed. Hence, $\mathcal A_X$ is weakly closed.

    By Step 4, we have $(H_{1,n},\ldots,H_{r,n}) \in \mathcal A_X$ for every $n$.
    Note that $0 \le H_{j,n} \le B$ holds $(\mathbb P\otimes\lambda)$-almost everywhere.
    Thus, the sequence $\left\{(H_{1,n},\ldots,H_{r,n})\right\}_{n\ge1}$ is bounded in the Hilbert space $(\mathbb L^2)^r$. Hence, there exist a subsequence $\{n_k\}_{k\ge1}$ and $(H_1,\ldots,H_r) \in (\mathbb L^2)^r$ such that $(H_{1,n_k},\ldots,H_{r,n_k}) \rightharpoonup (H_1,\ldots,H_r)$ weakly in $(\mathbb L^2)^r$.
    Since $\mathcal A_X$ is weakly closed, $(H_1,\ldots,H_r) \in \mathcal A_X$. In particular, $0 \le H_1 \le \cdots \le H_r \le X$ holds $(\mathbb P\otimes\lambda)$-almost everywhere and $H_j(\omega,\cdot)$ admits a nondecreasing representative for $\mathbb P$-almost every $\omega$.

    \textit{Step 6 (Identification)}.
    By Step 4, $H_{r,n_k} = X_{n_k}$. By Step 2, $X_{n_k}\to X$ strongly in $\mathbb L^2$, and hence weakly in $\mathbb L^2$.
    On the other hand, $H_{r,n_k} \rightharpoonup H_r$ weakly in $\mathbb L^2$.
    By uniqueness of weak limits, $H_r = X$ in $\mathbb L^2$.

    The expectation operator is linear and bounded, for every $j=1,\ldots,r$, $H_{j,n_k} \rightharpoonup H_j$ weakly in $\mathbb L^2$ implies that $\E{H_{j,n_k}(\cdot)} \rightharpoonup \E{H_{j}(\cdot)}$ weakly in $L^2((0,1),\lambda)$.
    On the other hand, Step 4 gives $\E{H_{j,n_k}(\cdot)} = h_{j,n_k}(\cdot)$, and Step 3 gives $h_{j,n_k}(\cdot) \to h_j(\cdot)$ strongly in $L^2((0,1),\lambda)$.
    Hence, $h_{j,n_k} \rightharpoonup h_j$ weakly in $L^2((0,1),\lambda)$. 
    By uniqueness of weak limits, $\E{H_j} = h_j$ in $L^2((0,1),\lambda)$.

    \textit{Step 7 (Canonical left-continuous versions)}.
    For $j=1,\ldots,r$, define
    \begin{align*}
        \varepsilon_m(q) = \min\left\{\frac{q}{2},\frac{1}{m}\right\},\quad q\in(0,1),
    \end{align*}
    and let
    \begin{align*}
        \widehat H_j(\omega,q)
        =
        \limsup_{m\to\infty}\frac{1}{\varepsilon_m(q)}\int_{q-\varepsilon_m(q)}^q H_j(\omega,s)\,ds.
    \end{align*}
    For every $m$, the map $(\omega,q)\mapsto (1/\varepsilon_m(q))\int_{q-\varepsilon_m(q)}^q H_j(\omega,s)\,ds$ is $\mathcal A \otimes\mathcal B((0,1))$-measurable. Hence, $\widehat H_j$ is jointly measurable.
    For $\mathbb P$-almost every $\omega$, let $\tilde{H}_j(\omega,\cdot)$ be a nondecreasing representative of $H_j(\omega,\cdot)$. Then $\widehat H_j(\omega,q) = \lim_{r\uparrow q} \tilde{H}_j(\omega,r)$ for $q\in(0,1)$.
    Thus, $\widehat H_j(\omega,\cdot)$ is nondecreasing and left-continuous.
    Moreover, $\widehat H_j(\omega,\cdot)$ differs from $\tilde{H}_j(\omega,\cdot)$ only at discontinuity points of $\tilde{H}_j$, which form a countable set. Therefore, $\widehat H_j = H_j$ holds $(\mathbb P\otimes\lambda)$-almost everywhere.
    Since $0\le H_j\le X\le B$ almost everywhere, $\widehat H_j(\omega,\cdot) \in \mathcal Q_{[0,B]}$ for $\mathbb P$-almost every $\omega$.
    Replacing $H_j$ by $\widehat H_j$ does not alter any of the preceding almost-everywhere equalities or inequalities. Thus, after this replacement, $0 \le H_1 \le \cdots \le H_r = X$ holds $(\mathbb P\otimes\lambda)$-almost everywhere.
    The Fubini theorem yields the desired result.
\end{proof}

\begin{proof}[Proof of Proposition \ref{prop:int}]
    \textit{Proof of part (i)}.
    It suffices to characterize the smallest possible lower endpoint and the largest possible upper endpoint over all observationally equivalent counterfactual processes satisfying Assumption \ref{assumption:I-MTR}, and to show that they are jointly attainable.

    Proposition \ref{prop:gen-simple}(i) implies
    \begin{align*}
        -\E{L_i(t)}
        &\le
        -\E{L_i\mathbf{1}\{D_i\le t\}} - K_0\P{D_i>t},\\
        \E{U_i(t)}
        &\le
        \E{U_i\mathbf{1}\{D_i\ge t\}} + K_1\P{D_i<t}.
    \end{align*}
    Hence, every interval in the sharp identified set is contained in $I_{\mathtt{MTR}}(t)$.

    It remains to show that the two endpoints are jointly attainable.
    Define
    \begin{align*}
        [\widetilde{L}_i(s),\widetilde{U}_i(s)]
        =
        \begin{cases}
            [L_i,U_i],
            & D_i<t,\ s<t,\\
            [L_i,K_1],
            & D_i<t,\ s\ge t,\\
            [L_i,U_i],
            & D_i=t,\ s\in\mathcal T,\\
            [K_0,U_i],
            & D_i>t,\ s\le t,\\
            [L_i,U_i],
            & D_i>t,\ s>t.
        \end{cases}
    \end{align*}
    The constructed interval belongs to $\mathcal Y$ for every $s\in\mathcal T$, is observationally equivalent, and satisfies MTR.
    Under this construction,
    \begin{align*}
        \mathbb{E}_{\oplus}\left[[\widetilde{L}_i(t),\widetilde{U}_i(t)]\right] 
        =
        \bigg[
            \E{L_i\mathbf{1}\{D_i\le t\}} + K_0\P{D_i>t},
            \,
            \E{U_i\mathbf{1}\{D_i\ge t\}} + K_1\P{D_i<t}
        \bigg].
    \end{align*}
    Hence, $I_{\mathtt{MTR}}(t)$ belongs to the sharp identified set.\hfill $\square$

    \noindent\textit{Proof of part (ii)}.
    A similar reasoning to above shows that every element of the sharp identified set is contained in $I_{\mathtt{MTS}}(t)$. Thus we only show that the bounds are jointly attainable.
    For each $s\in\mathcal T$, define the counterfactual interval by
    \begin{align*}
        [\widetilde{L}_i(s),\widetilde{U}_i(s)]
        =
        \begin{cases}
            \left[
                K_0,\,
                \E{U_i\mid D_i=s}
                \right],
                & D_i<s,\\
                [L_i,U_i],
                & D_i=s,\\
                \left[
                \E{L_i\mid D_i=s},\,
                K_1
            \right],
            & D_i>s.
        \end{cases}
    \end{align*}
    Each constructed interval belongs to $\mathcal Y$ and observationally equivalent.
    We verify MTS.
    For every $s,d\in\mathcal T$,
    \begin{align*}
        \mathbb{E}[\widetilde L_i(s)\mid D_i=d]=
        \begin{cases}
            K_0,
            & d<s,\\
            \E{L_i\mid D_i=s},
            & d\ge s,
        \end{cases}
    \end{align*}
    which is weakly increasing in $d$. The upper bound can be treated similarly. Hence, the constructed counterfactual process satisfies MTS.
    Under this construction, a direct computation yields $\mathbb{E}_{\oplus}[[\widetilde{L}_i(t),\widetilde{U}_i(t)]] = I_{\mathtt{MTS}}$, that is, $I_{\mathtt{MTS}}(t)$ belongs to the sharp identified set.\hfill $\square$

    \noindent\textit{Proof of part (iii)}.
    A similar reasoning to above shows that every element of the sharp identified set is contained in $I_{\mathtt{MIV}}(t)$. Thus we only show that the bounds are jointly attainable.
    For each $z\in\mathcal Z$, define
    \begin{align*}
        \alpha_z
        &=
        \E{L_i\mathbf 1\{D_i=t\}\mid Z_i=z} + K_0\P{D_i\ne t\mid Z_i=z},\\
        \beta_z
        &=
        \E{U_i\mathbf 1\{D_i=t\}\mid Z_i=z} + K_1\P{D_i\ne t\mid Z_i=z},
    \end{align*}
    and let $a_z^* = \sup_{v\le z}\alpha_v$ and $b_z^* = \inf_{v\ge z}\beta_v$.
    By construction, both $z\mapsto a_z^*$ and $z\mapsto b_z^*$ are weakly increasing.
    Because MIV assumption is maintained, there exists an observationally equivalent counterfactual process $\{[L_i^0(s),U_i^0(s)]:s\in\mathcal T\}$ satisfying Assumption \ref{assumption:I-MIV}. 
    Define $a_z^0 = \E{L_i^0(t)\mid Z_i=z}$ and $b_z^0 = \E{U_i^0(t)\mid Z_i=z}$. Then $z\mapsto a_z^0$ and $z\mapsto b_z^0$ are weakly increasing.
    Moreover, notice that
    \begin{align*}
        \alpha_z \le a_z^0 \le 
        \E{L_i\mathbf 1\{D_i=t\}\mid Z_i=z} + K_1\P{D_i\ne t\mid Z_i=z}
    \end{align*}
    and 
    \begin{align*}
        \E{U_i\mathbf 1\{D_i=t\}\mid Z_i=z} + K_0\P{D_i\ne t\mid Z_i=z}
        \le b_z^0 \le \beta_z.
    \end{align*}
    We then obtain that
    \begin{align*}
        \alpha_z 
        \le
        \sup_{v\le z}\alpha_v
        \le
        \sup_{v\le z}a_v^0
        =
        a_z^0,\quad 
        b_z^0
        =
        \inf_{v\ge z}b_v^0
        \le
        \inf_{v\ge z}\beta_v
        \le
        \beta_z,
    \end{align*}
    and
    \begin{align*}
        \inf_{v\ge z}\beta_v-\sup_{v\le z}\alpha_v
        &\ge
        b_z^0-a_z^0
        =
        \E{(U_i^0(t)-L_i^0(t)) \mid Z_i=z}\\
        &\ge
        \E{(U_i^0(t)-L_i^0(t)) \mathbf{1}\{D_i=t\} \mid Z_i=z}
        =
        \E{(U_i-L_i)\mathbf{1}\{D_i=t\}\mid Z_i=z}.
    \end{align*}
    %For every $z$ such that $\P{D_i\ne t\mid Z_i=z}>0$, define
    Define
    \begin{align*}
        \ell_z^*
        =
        \frac{\sup_{v\le z}\alpha_v-\E{L_i\mathbf 1\{D_i=t\}\mid Z_i=z}}{\P{D_i\ne t\mid Z_i=z}},\quad
        u_z^*
        =
        \frac{\inf_{v\ge z}\beta_v-\E{U_i\mathbf 1\{D_i=t\}\mid Z_i=z}}{\P{D_i\ne t\mid Z_i=z}}.
    \end{align*}
    The preceding inequalities imply that $K_0\le \ell_z^*\le u_z^*\le K_1$.
    %For $z$ such that $\P{D_i\ne t\mid Z_i=z}=0$, set $(\ell_z^*,u_z^*)=(K_0,K_1)$; the choice on such cells is immaterial.
    Now define a new counterfactual process by
    \begin{align*}
        [\widetilde{L}_i(s),\widetilde{U}_i(s)]
        =
        \begin{cases}
            [L_i^0(s),U_i^0(s)],
            & s\ne t,\\
            [L_i,U_i],
            & s=t,\ D_i=t,\\
            [\ell_{Z_i}^*,u_{Z_i}^*],
            & s=t,\ D_i\ne t.
        \end{cases}
    \end{align*}
    For every $s\in\mathcal T$, the constructed interval belongs to $\mathcal Y$ and observationally equivalent. 
    For $s\ne t$, MIV is inherited from the original process.
    For $s=t$, the construction gives $\mathbb{E}[\widetilde L_i(t)\mid Z_i=z] = \sup_{v\le z}\alpha_v$ and $\mathbb{E}[\widetilde U_i(t)\mid Z_i=z] = \inf_{v\ge z}\beta_v$. Since both $z\mapsto \sup_{v\le z}\alpha_v$ and $z\mapsto \inf_{v\ge z}\beta_v$ are weakly increasing, the constructed counterfactual process satisfies Assumption \ref{assumption:I-MIV}.
    Finally,
    \begin{align*}
        \mathbb{E}[\widetilde L_i(t)]=\sum_{z\in\mathcal Z}\P{Z_i=z} \sup_{v\le z}\alpha_v,\quad \mathbb{E}[\widetilde U_i(t)] = \sum_{z\in\mathcal Z}\P{Z_i=z} \inf_{v\ge z}\beta_v.
    \end{align*}
    Hence, $\mathbb{E}_\oplus[\widetilde{Y}_i(t)] = I_{\mathtt{MIV}}(t)$, and $I_{\mathtt{MIV}}(t)$ belongs to the sharp identified set.\hfill $\square$

    \noindent\textit{Proof of part (iv)}.
    A similar reasoning to above shows that every element of the sharp identified set is contained in $I_{\mathtt{MTR+MTS}}(t)$. Thus we only show that the bounds are jointly attainable.
    For $s\le d$, define
    \begin{align*}
        \lambda_{d,s}^L
        =
        \begin{cases}
            \dfrac{\E{L_i\mid D_i=s}-K_0}{\E{L_i\mid D_i=d}-K_0} & \text{if }\E{L_i\mid D_i=d}>K_0,\\
            1 & \text{if }\E{L_i\mid D_i=d}=K_0,
        \end{cases}
    \end{align*}
    and, for $s\ge d$, define
    \begin{align*}
        \lambda_{d,s}^U
        =
        \begin{cases}
            \dfrac{K_1-\E{U_i\mid D_i=s}}{K_1-\E{U_i\mid D_i=d}} & \text{if }\E{U_i\mid D_i=d}<K_1,\\
            1 & \text{if }\E{U_i\mid D_i=d}=K_1.
        \end{cases}
    \end{align*}
    On the event $\{D_i=d\}$, define
    \begin{align*}
    \widetilde L_i(s)
    =
    \begin{cases}
        K_0+\lambda_{d,s}^L(L_i-K_0) & s\le d,\\
        L_i & s>d,
    \end{cases}\quad
    \widetilde U_i(s)
    \coloneqq
    \begin{cases}
        U_i & s<d,\\
        K_1-\lambda_{d,s}^U(K_1-U_i) & s\ge d.
    \end{cases}
    \end{align*}
    Since $\lambda_{d,d}^L=\lambda_{d,d}^U=1$, the construction is observationally equivalent.
    For $s\le d$, we have $\widetilde L_i(s)\le L_i\le U_i=\widetilde U_i(s)$, whereas for $s\ge d$, $\widetilde L_i(s)=L_i\le U_i\le\widetilde U_i(s)$, since $\lambda_{d,s}^L, \lambda_{d,s}^U \in[0,1]$, which is because $\E{L_i\mid D_i=s}$ and $\E{U_i\mid D_i=s}$ are increasing in $s$ under MTR and MTS.
    Hence, $[\widetilde L_i(s),\widetilde U_i(s)]\in\mathcal Y$.
    Moreover, $s\mapsto\lambda_{d,s}^L$ is weakly increasing on $\{s:s\le d\}$, whereas $s\mapsto\lambda_{d,s}^U$ is weakly decreasing on $\{s:s\ge d\}$. Therefore, $\widetilde L_i(s)$ and $\widetilde U_i(s)$ are weakly increasing, so that the constructed process satisfies Assumption \ref{assumption:I-MTR}.
    Furthermore, for every $s,d\in\mathcal T$,
    \begin{align*}
        \mathbb{E}[\widetilde L_i(s)\mid D_i=d] = \E{L_i \mid D_i = s\wedge d},\quad
        \mathbb{E}[\widetilde U_i(s)\mid D_i=d] = \mathbb{E}[U_i \mid D_i = s\vee d].
    \end{align*}
    Therefore, since $\E{L_i\mid D_i=d}$ and $\E{U_i\mid D_i=d}$ are increasing in $d$, so are $\mathbb{E}[\widetilde L_i(s)\mid D_i=d]$ and $\mathbb{E}[\widetilde U_i(s)\mid D_i=d]$, implying that the constructed process also satisfies Assumption \ref{assumption:I-MTS}.
    Finally, at the target treatment $t$,
    \begin{align*}
        \mathbb{E}[\widetilde L_i(t)]&=\sum_{d<t}\mathbb{E}[L_i \mid D_i=d]\P{D_i=d} + \mathbb{E}[L_i \mid D_i=t]\P{D_i\ge t},\\
        \mathbb{E}[\widetilde U_i(t)]
        &=\sum_{d>t}\mathbb{E}[U_i \mid D_i=d]\P{D_i=d} + \E{U_i \mid D_i=t}\P{D_i\le t}.
    \end{align*}
    Hence, $\mathbb{E}_\oplus[\widetilde{Y}_i(t)] = I_{\mathtt{MTR+MTS}}(t)$, and $I_{\mathtt{MTR+MTS}}(t)$ belongs to the sharp identified set.
\end{proof}
\bibliographystyle{apalike} 
\bibliography{refs}

\end{document}